\documentclass[envcountsame]{llncs}
\usepackage{macros}
\setboolean{extended}{true}
\begin{document}
\title{Visualising CTL Witnesses and Counterexamples
\ifextended --- Extended Version
\fi}
\author{Arend Rensink\orcidID{0000-0002-1714-6319}}
\institute{University of Twente, \href{mailto:arend.rensink@utwente.nl}{\email{arend.rensink@utwente.nl}}}
\maketitle

\begin{abstract}
One of the advantages of \LTL over \CTL is that the notion of a counterexample is easy to grasp, visualise and process: it is a trace that violates the property at hand. In this paper we propose a notion of \emph{evidence} for \CTL properties on explicit-state models --- which equally serves as witness for satisfied properties and counterexample for violated ones --- and how to visualise it, with the main aim of (human) comprehension. The main contribution consists of a formal model of evidence, a characterisation of minimal evidence per temporal operator, and a concrete, implemented proposal for its visualisation.

\ifextended\smallskip
This is the extended version of \cite{CTLviz-SPIN}, containing the proofs of all results.
\fi
\end{abstract}


\section{Introduction}
\label{sec:intro}

Temporal logic is used to express properties of the evolution of a system over time, such as ``a certain property $\phi$ holds in the system until another property $\psi$ starts to hold'', or ``provided $\phi$ keeps occurring again and again, $\psi$ is never reachable''. The family of temporal logics consists of two main branches, which are incomparable in the sense that either can express properties that the other cannot: \emph{Linear Temporal Logic} (\LTL, ``linear time'')
%, originally proposed by Pnueli 
\cite{LTL}, and \emph{Computational Tree Logic} (\CTL, ``branching time'')
%, originally proposed by Clark and Emerson 
\cite{CTL1,CTL2}. A more powerful logic that combines the features of these two is \CTLstar
%(Emerson and Halpern 
\cite{CTL*}%)
, and a yet more expressive variant is the modal $\mu$-calculus 
%(credited to Scott and De Bakker but first published by Kozen 
\cite{mu1,mu2}%)
.

Much has been said about the relative merits of \LTL and \CTL, regarding expressivenes, computational complexity of model checking (i.e., the process of finding out whether or not a temporal property holds in a given system model), and usability in practice; see, e.g., 
%Vardi 
\cite{CTL-vs-LTL}. In this paper, we address a relatively minor difference, namely the kind of information we can glean from the satisfaction, respectively violation of a temporal property by a given model --- besides the bare, boolean fact itself. Such knowledge is naturally generated as a by-product of the model checking process --- or in sometimes, as in \cite{Counterexample-driven-MC}, is actually central to the model checking algorithm itself.
One of the reasons why \LTL is found to be useful is that, if a property $\phi$ is \emph{not} satisfied, there is at least one (possibly infinite) trace that violates $\phi$, called a \emph{counterexample} to $\phi$. Counterexamples are very useful because they can be further analysed, in order to understand how the violation of $\phi$ came about. Another scenario is where one actually wants to find a concrete trace that \emph{does} satisfy some \LTL property $\phi$: this can be obtained by checking $\neg\phi$ --- a counterexample then serves as a trace of the desired kind.

In \CTL, there is no such simple notion of counterexample. This is due to the nature of the logic, which (as the respective adjectives ``branching'' and ``linear'' imply) takes the model's branching structure into account, rather than essentially regarding it as a set of (finite and infinite) linear traces. In general, the non-satisfaction of a property can therefore only be fully understood by taking at least part of that branching structure aboard, making this conceptually more complex than a single trace or even a set of them. What part of the branching structure, and how to capture it, are questions we address in this paper.

\medskip\noindent Another difference between \LTL and \CTL we want to recall is that violating a \CTL property is equivalent to satisfying its negation. That is not the case for \LTL: e.g., considering the property $\phi=\F p$, which is satisfied by every trace that reaches a state where $p$ holds, we can easily imagine a system where neither $\phi$ nor $\neg\phi$ is satisfied: it merely has to include one trace leading to a $p$-state, and a second one that loops without every reaching such a $p$-state. The closest alternative in \CTL would be $\EF p$ (``there is a trace that reaches a state where $p$ holds''), but if this is violated by a given system then its negation $\neg\EF p$ (being equivalent to $\AG \neg p$, ``for all traces, $\neg p$ holds all along'') is certainly satisfied.

As a consequence, in the case of \CTL there is actually no reason to focus on counterexamples for violated properties: we can simultaneously and equivalently consider the dual notion of \emph{witnesses} for property satisfaction. As an overarching term, we use \emph{evidence} for a property $\phi$ on a given model $\M$: if $\M$ satisfies $\phi$, the evidence is a witness of that fact, whereas if $\M$ violates $\phi$, the evidence is a counterexample.

\medskip\noindent The overall question addressed by this paper is twofold:
\begin{inumerate}
\item what is the simplest notion of evidence that precisely captures the reason for satisfaction or violation of an arbitrary \CTL formula, and
\item\label{RQ2} how can we effectively visualise this notion?
\end{inumerate}
With respect to~\ref{RQ2}, we are obviously limited by the fact that visualisation inherently stops being useful for human comprehension for any but small models.

After setting up the definitions in \cref{sec:definitions}, the main technical contribution of this paper is given in \cref{sec:evidence}. \Cref{sec:proofs,sec:combined} are mostly concerned with visualisation. \Cref{sec:related,sec:conclusion} wrap up with related work and conclusions.
\ifextended \Cref{sec:appendix} contains the proofs of all results.
\else

An extended version of this paper that includes the proofs of all results can be found at \cite{CTLviz-extended}.
\fi

\section{Definitions}
\label{sec:definitions}

We start off by presenting the concepts and notations used in this paper. Most of them are standard, but we sometimes adopt alternative formulations to serve the flow of the paper. The only technical novelties are those of \emph{closed states} and \emph{submodel} (\cref{def:model,def:submodel}).

We use directed graphs $\tupof{X,R}$, where $X$ is countable and $R\subseteq X\times X$. Such a graph is called \emph{cyclic} if there are $\setof{(x_i,x_{i+1})}_{1\leq i\leq n}\subseteq R$ for some $n\geq 1$ such that $x_{n+1}=x_1$, and \emph{acyclic} otherwise.

We use finite and infinite sequences over a set $X$, denoted $X\finseq$ and $X\infseq$, as well as $X\finfseq=X\finseq\cup X\infseq$. We write $|\sigma|\in\Nat\cup\setof\omega$ to denote the length of $\sigma\in X\finfseq$, $\sigma\proj i$ for its $i$th element (for any $i\in\Nat$ such that $1\leq i\leq|\sigma|$) and $\sigma\last=\sigma\proj{|\sigma|}$ for its last element (if $\sigma\in X\finseq$). We use $[\sigma]=\gensetof{\sigma\proj i}{1\leq i\leq |\sigma|}$ for the set of elements in $\sigma$, ${<_\sigma}= \gensetof{(\sigma\proj i,\sigma\proj{i+1})}{1\leq i<|\sigma|}$ for the induced ordering, and $\tupof\sigma=\tupof{[\sigma],{<_\sigma}}$ for its graph. $\sigma$ is called acyclic if $\tupof\sigma$ is so. We write $\emptystr$ for the empty sequence, $\sigma_1\cc \sigma_2$ for the concatenation of $\sigma_1$ and $\sigma_2$ (also extended pointwise to sets of sequences), and $\sigma\prfeq \sigma'$ to indicate that $\sigma$ is a prefix of $\sigma'$ (meaning that $\sigma'=\sigma\cc\sigma''$ for some $\sigma''$).

We write $f\of X\pto Y$ to denote that $f$ is a partial function from $X$ to $Y$; $\dom f\subseteq X$ then denotes its domain and $\hat f=\gensetof{(x,f(x))}{s\in\dom f}$ its graph;\footnote{This is the common term in this context; however, everywhere else in this paper we will exclusively use \emph{graph} for structures $\tupof{X,R}$ as introduced earlier.} moreover, $f\proj Z\of X\pto Y$ for $Z\subseteq X$ denotes the restriction of $f$ to the elements in $Z$. The set of booleans is denoted $\Bool$, with elements $\fff,\ttt$.

\medskip\noindent
To build formulas, we assume a set of \emph{basic propositions} $\Prop$ and a set of \emph{logical operators} $\Oper$. Every operator $o\in \Oper$ has an \emph{arity} $\alpha_o\in \Nat$. For instance, \PL (propositional logic) has basic operators $\Oper_\PL=\setof{\True^{(0)},\False^{(0)},{\lneg}^{(1)},{\land}^{(2)},{\lor}^{(2)}}$ (where a superscript $(i)$ indicates arity $i$),\footnote{We use programming-like notation for the logical operators to be consistent with the visualisation; in particular, $\land$ stands for conjunction and $\lor$ for disjunction.} and \CTL (Computation Tree Logic) has $\Oper_\CTL= \Oper_\PL \cup \setof{\AX^{(1)},\EX^{(1)},\AF^{(1)},\EF^{(1)},\AG^{(1)},\EG^{(1)},\EU^{(2)},\AU^{(2)}}$.

From propositions and operators we build \emph{formulas}, as a set $\Form$ such that $\Form\subseteq \Prop\cup (\Oper\times\Form\finseq))$.
\begin{definition}[formula]\label{def:formula}
The set of \emph{formulas} $\Form$ is the smallest set such that $\Prop\subseteq \Form$ and $\gensetof{(o,\bar\phi)}{o\in \Oper,\bar\phi\in \Form\finseq, |\phi|=\alpha_o}\subseteq \Form$.
%
%The set of \emph{indexed formulas} is defined as $\IxForm=\Nat\finseq\times \Form$. Given two indexed formulas $\iota_i=(\sigma_i,\phi_i)\in \IxForm$ for $i=1,2$, $\iota_1$ is the \emph{parent} of $\iota_2$ if $\sigma_2=\sigma_1\cc k$ and $\phi_1=(o,\psi_1\cdots\psi_n)$ with $\phi_2=\psi_k$ for some $k\in \Nat$.
\end{definition}
We use $\child(\phi)$ to denote the set of arguments of $\phi$, setting $\child(p)=\emptyset$ for all $p\in \Prop$. For any set of formulas $F\subseteq \Form$, we use $F_\Prop=F\cap \Prop$ to denote the subset of propositions, and $F_\Oper=F\cap (\Oper\times \Form\finseq)$ to denote the subset of composite formulas.
A set of formulas $F$ is called \emph{subformula-closed} if $\phi\in F$ implies $\child(\phi)\subset F$. The subformula-closed sets of formulas are collected in $\CForm$.

\medskip\noindent The semantics of \CTL is usually defined in terms of Kripke structures \cite{Kripke}, in which the states (taken from some universe $\State$) are labelled with the sets of propositions that hold there. In this paper we use a variant, in which we add a subset of \emph{closed} states (which start playing a role in \cref{def:submodel}), and use a more general labelling function that 
\begin{inumerate}
\item is partial, and
\item includes arbitrary formulas (and not just propositions).
\end{inumerate}
Hence, labellings are partial functions $L\of S\times F\pto \Bool$ for some set of states $S$ and set of formulas $F$. We use $L_\Prop=L\cap (S\times F_\Prop)$ and $L_\Oper=L\cap (S\times F_\Oper)$ to denote the subfunctions for propositional and composed formulas.

\begin{definition}[model, path]\label{def:model}
Let $F\in \CForm$. A \emph{model over $F$} is a tuple $M=\tupof{S,C,R,L}$ where $S$ is a set of states with subset $C\subseteq S$ of \emph{closed} states, $R\subseteq S\times S$ a binary relation over $S$ and $L\of S\times F\pto \Bool$ a labelling. $M$ is called \emph{Kripke} if $C=S$ and $L=L_\Prop$, and \emph{full} if $C=S$ and $L$ is total.

A \emph{path} in $M$ is a sequence $\rho\in S\finfseq$ such that $(\rho\proj i,\rho\proj{i+1})\in R$ for all $1\leq i< |\rho|$. $\rho$ is called \emph{maximal} if there is no path of which it is a proper prefix.
\end{definition}
Note that paths may contain repeated occurrences of states. The set of models over $F$ is denoted $\Model(F)$. We commonly write $s\to s'$ for $(s,s')\in R$. We use $\maxpath(s)$ for all maximal paths starting in $s$, $\succ(s)=\gensetof{s'}{s\to s'}$ to denote the set of direct successors of $s$, and $S\proj{\phi,v}\subseteq S$ (with $v\in\Bool$) to denote the subset of states $s$ for which $L(s,\phi)=v$, and analogous for $C$.

Given $F\in\CForm$ and an arbitrary state set $S$, every labelling $L\of S\times F\pto \Bool$ gives rise to a ``discrete'' model $\tupof L = \tupof{S,\emptyset,\emptyset,L}\in \Model(F)$.

\begin{example}\label{ex:game}
As a running example, we use a very simple (and stupid) single-player game, played on the board depicted in \Cref{fig:board}. The rules are that, as long as the game hasn't ended, the player alternates between throwing a two-sided die and moving the pawn forward for the indicated number of eyes. Landing on an \textsf{L}-field ends the game with a loss; landing on \textsf{W} ends it with a win.
\end{example}
\vspace*{-\topsep}
\begin{wrapfigure}{r}{.3\textwidth}\centering
\vspace*{-3mm}
\includegraphics[scale=.5]{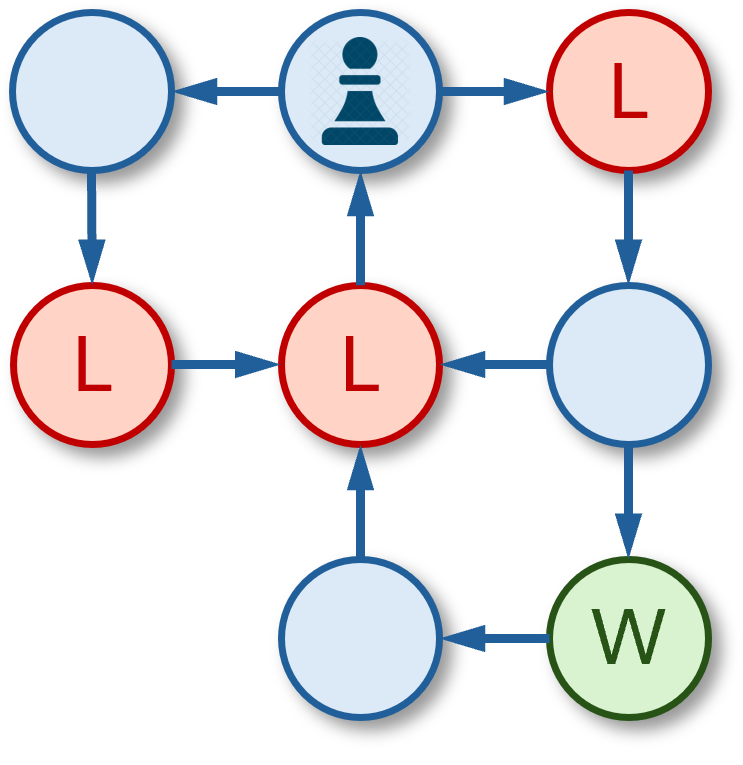}
\vspace*{-3mm}
\caption{Game board (\textsf{L} = loss, \textsf{W} = win)}
\vspace*{-3mm}
\label{fig:board}
\end{wrapfigure}

\Cref{fig:state-space} shows the state space of this game in the form of a Kripke model. Transitions are labelled for comprehensibility, but their labels play no role in the formal model. Propositions are shown in green if they are satisfied (i.e., labelled $\ttt$), and in red otherwise.

For this game, one might be interested in questions such as \emph{can the player win without throwing 1?} or \emph{does there exist a maximal path where winning always remains possible but is never actually reached?} These can be formalised as \CTL properties; respectively $\EUof{\,\lneg\done\,}{\win}$ and $\EG( \lneg\win\land \EF\win)$.
\qed
\begin{figure}[tb]\centering
\includegraphics[scale=.22]{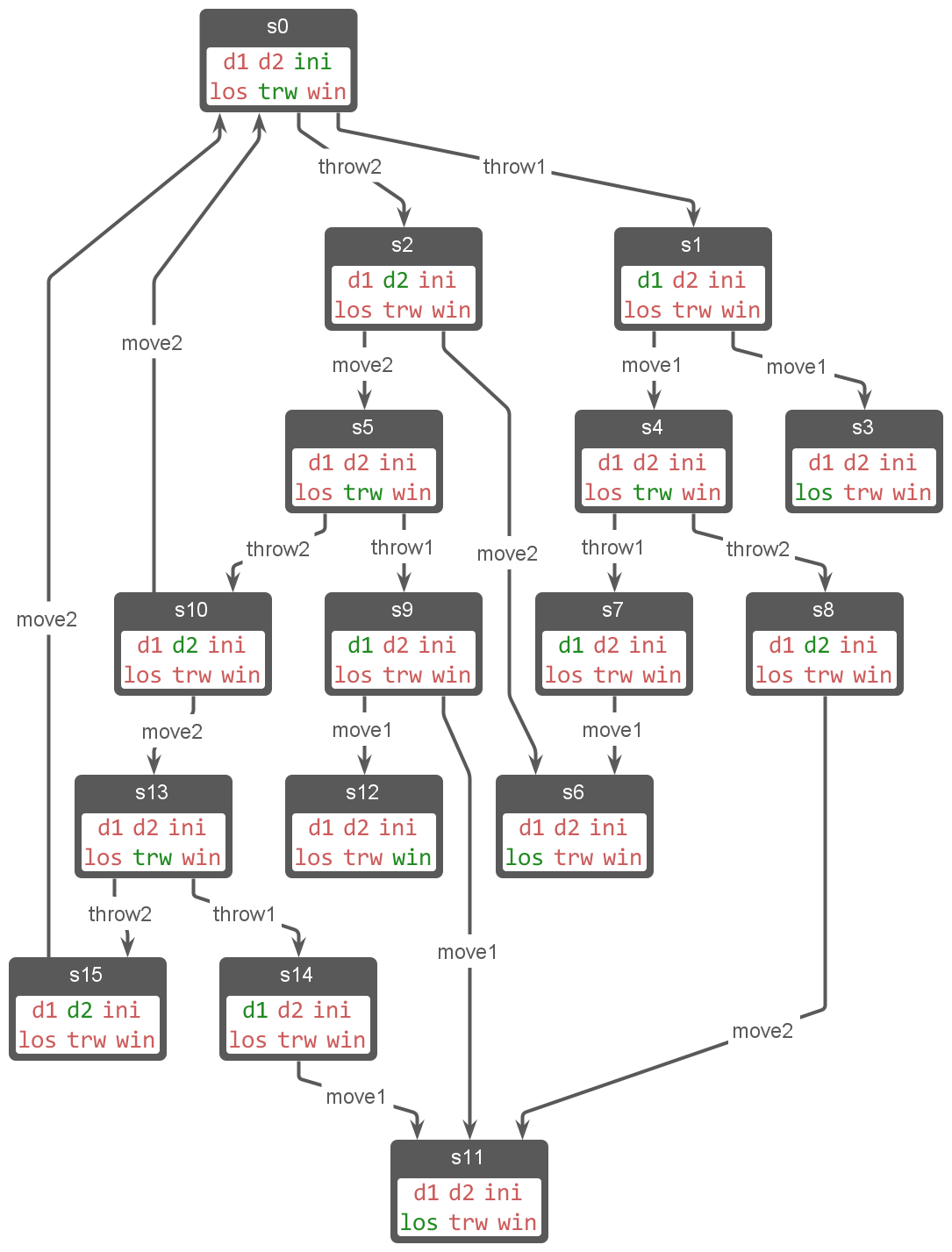}
\caption{Kripke model of the game state space (cf.\ \Cref{ex:game}). $\Prop$ consists of \done (the die shows 1 and the pawn should move), \dtwo (the die shows 2 and the pawn should move), \ini (the pawn is in the initial field), \los (the game is lost), \trw (the die should be thrown) and \win (the game is won). Propositions are green where satisfied, red elsewhere.}
\label{fig:state-space}
\end{figure}

\medskip\noindent
We now come to the standard semantics of \CTL, which is captured by a family of relations ${\sat^M}\subseteq S\times F$ for Kripke models $M$, defined as follows (limited to the existential fragment consisting of operators $\setof{\True, {\lneg}, {\lor}, \EX, \EU, \EG}$ as all others can be generated from those\footnote{See, e.g., \url{https://en.wikipedia.org/wiki/Computation_tree_logic}}):
\begin{align*}
s & \sat^M p
  && \text{if } L(s,p)=\ttt \\
s & \sat^M \True
  && \text{always} \\
s & \sat^M \lneg\phi
  && \text{if } s\not\sat^M \phi \\
s & \sat^M \phi_1\lor \phi_2
  && \text{if } s\sat^M \phi_1 \text{ or } s\sat^M \phi_2 \\
s & \sat^M \EX\phi
  && \text{if there is a } s'\in\succ(s) \text{ such that } s'\sat^M \phi \\
s & \sat^M \EUof{\phi_1}{\phi_2}
  && \begin{array}[t]{l}
     \text{if there is a } \rho\in \maxpath(s) \text{ and } n < |\rho| \\
     \text{such that } \rho\proj i\sat^M \phi_1 \text{ for all } 1\leq i\leq n \text{ and } \rho\proj{n+1} \sat^M \phi_2
     \end{array} \\
s & \sat^M \EG\phi
  && \begin{array}[t]{l}
     \text{if there is a } \rho\in\maxpath(s) \\
     \text{such that } \rho\proj {i+1}\sat^M \phi \text{ for all } 0\leq i<|\rho|
     \end{array}
\end{align*}
For instance, checking the properties of \Cref{ex:game} on the Kripke model of \Cref{fig:state-space} yields $s_0\not\sat \EUof{\,\lneg\done\,}\win$ (from the initial state, winning is impossible without throwing 1) and $s_0\sat \EG(\lneg\win \land \EF\win)$ (by staying in the loop $s_0\,s_2\,s_5\,s_{10}\,s_0$, the player does not win, but the win state $s_{12}$ always remains reachable.)

\medskip\noindent The distinction between open and closed states comes into play when we consider \emph{submodels}. A submodel of $M$ consists of subsets of $S_M$, $C_M$, $R_M$ and $L_M$, with the proviso that a closed state of the submodel must already have all outgoing transitions that it has in $M$. Formally:
\begin{definition}[submodel]\label{def:submodel}
Let $F\in \CForm$, and let $M_1,M_2\in \Model(F)$. $M_1$ is a \emph{submodel} of $M_2$, denoted $M_1\submodel M_2$, if $S_1\subseteq S_2$, $C_1\subseteq C_2$, $R_1\subseteq R_2$ and $L_1\subseteq L_2$, and $(s,s')\in R_2\setminus R_1$ implies $s\notin C_1$.
\end{definition}
Among other things, $M_1\submodel M_2$ implies that none of the ``new'' states in $S_2\setminus S_1$ are reachable from any of the closed ``old'' states in $C_1$.
\begin{definition}[soundness, consistency]
Let $F\in \CForm$. A model $M\in \Model(F)$ is called \emph{sound} if it is
full and for all $(s,\phi)\in S\times F$, $L(s,\phi)=\ttt$ iff $s\sat^M \phi$ (and conversely,  $L(s,\phi)=\fff$ iff $s\not\sat^M \phi$). $M$ is called \emph{consistent} if it has a sound supermodel.
\end{definition}
Note that every Kripke model in which $\dom L=S\times F_\Prop$ is consistent. Hence, the traditional task of model checking a formula $\phi$ over a given Kripke model $M$ comes down to finding the smallest sound supermodel $N\supmodel M$ (where the qualification ``smallest'' implies that $N$ is unique and has $S_N=S_M$, $C_N=C_M$ and $R_N=R_M$). Instead, in this paper, we rather focus on the equivalent task of establishing whether (or actually: why) a given full model is sound.

\section{Evidence}
\label{sec:evidence}

Given a labelling function $L\of S\times F\pto \Bool$, we refer to the elements of $\hat L$ (which have the shape $(s,\phi,b)$) as \emph{assertions}, ranged over by $a$. Such an assertion should be read as the claim that $\phi$ holds in $s$ if $b=\ttt$, or fails to hold in $s$ if $b=\fff$. For propositions ($\phi\in\Prop$), assertions are axiomatic --- they cannot be (dis)proved but must be taken at face value. For compound formulas ($\phi\in \Form_\Oper$) however, the correctness of an assertion depends on the subformulas of $\phi$.
We define \emph{evidence} for an assertion involving $\phi$ as sufficient information about the children of $\phi$ to guarantee that the assertion is correct.
\begin{definition}[evidence, witness, counterexample]\label{def:evidence}
Let $F\in \CForm$ and let $M\in \Model(F)$ be consistent. $M$ is \emph{evidence} for an assertion $a\in S\times F_\Oper\times \Bool$ if
\begin{enumerate*}[label=(\roman*)]
\item $\dom L_M\subseteq S\times \child(\phi_a)$ and 
\item $a\in \hat L_N$ for every sound supermodel $N\supmodel M$.
\end{enumerate*}
Moreover, $M$ is called a \emph{witness} for $(s,\phi)\in S\times F$ if it is evidence for $(s,\phi,\ttt)$, and a \emph{counterexample} for ($s,\phi)$ if it is evidence for $(s,\phi,\fff)$.
\end{definition}
For instance, starting with a sound model $M$ and an arbitrary assertion $a\in \hat L$, a straightforward case of evidence is obtained by just taking the transition structure of $M$ and reducing the labelling to the children of $\phi_a$.
\begin{restatable}{proposition}{propevidence}\label{prop:evidence}
Let $F\in \CForm$ and let $M\in \Model(F)$ be sound. For every $a\in \hat L$ with $\phi_a\in F$, the model $\tupof{S,C,R,L\proj{S\times \child(\phi_a)}}$ is evidence for $a$.
\end{restatable}
Evidence may easily contain information that is not strictly required to explain the assertion. This is already exemplified by \Cref{prop:evidence}: the model $M_a$ shown there to be evidence for an assertion $a$ contains the entire transition structure of the original model $M$, including states that are not even reachable from $s_a$. In the remainder of this section, we obtain more precise characterisations of evidence for the core \CTL formulas.

As a first step, \Cref{tab:evidence} defines concrete ``syntactic'' conditions $\witness$ and $\counter$ that are sufficient and, for a certain class of formulas, necessary to ensure that a model is evidence for a given assertion.
\begin{table}[t]
\caption{Witness and counterexample conditions $\witness_M,\counter_M\of S\times F_\Oper\to \Bool$ for a given model $M\in \Model(F)$.}
\label{tab:evidence}
\vspace*{-3mm}
\[\begin{array}{l@{{\enspace}}l}
\phi & \witness_M(s,\phi) \\
\hline
\True & \text{always} \\
\neg\psi & s\in S\proj{\psi,\fff} \\
\psi_1\lor\psi_2 & s\in S\proj{\psi_1,\ttt}\cup S\proj{\psi_2,\ttt} \\
\EX\psi & \succ(s) \text{ overlaps with } S\proj{\psi,\ttt} \\
\EUof{\psi_1}{\psi_2} & \maxpath(s) \text{ overlaps with } S\proj{\psi_1,\ttt}\finseq\cc S\proj{\psi_2,\ttt}\cc S\finfseq \\
\EG{\psi} & \maxpath(s) \text{ overlaps with } (S\proj{\psi,\ttt}\finseq\cc C\proj{\psi,\ttt}) \cup S\proj{\psi,\ttt}\infseq \\[\bigskipamount]
\phi & \counter_M(s,\phi) \\
\hline
\True & \text{never} \\
\neg\psi & s\in S\proj{\psi,\ttt} \\
\psi_1\lor\psi_2 &  s\in S\proj{\psi_1,\fff}\cap S\proj{\psi_2,\fff} \\
\EX\psi &  s\in C \text{, and } \succ(s) \text{ is a subset of } S\proj{\psi,\fff} \\
\EUof{\psi_1}{\psi_2} & \maxpath(s)\text{ is a subset of }C\proj{\psi_2,\fff}\finfseq\cup \bigl(C\proj{\psi_2,\fff}\finseq\cc (S\proj{\psi_1,\fff}\cap S\proj{\psi_2,\fff})\cc S\finfseq\bigr) \\
\EG{\psi} & \maxpath(s)\text{ is a subset of }C\finseq \cc S\proj{\psi,\fff}\cc S\finfseq. %\\
\end{array}\]
\end{table}
We first state their sufficiency.

\begin{samepage}
\begin{restatable}[$\witness$ and $\counter$ are sufficient]{theorem}{thmevidenceif}\label{thm:evidence-if}
Let $F\in \CForm$, let $M\in \Model(F)$ be consistent and let $s\in S$ and $\phi\in F_\Oper$.
\begin{enumerate}[smallsep]
\item $M$ is a witness for $(s,\phi)$ if $\witness_M(s,\phi)$ holds;
\item $M$ is a counterexample for $(s,\phi)$ if $\counter_M(s,\phi)$ holds.
\end{enumerate}
\end{restatable}
\end{samepage}
In general, $\witness$ and $\counter$ are not necessary for $M$ to be evidence: there exist witnesses that do not satisfy $\witness$, and counterexamples that do not satisfy $\counter$. In a sense, this is not surprising. Consider the case where $\phi$ is a valid formula (i.e., satisfied by any state of any sound model); then \emph{every} consistent model with $\dom L\subseteq S\times \child(\phi)$ is a witness for $(s,\phi)$ (for any $s$), even the trivial, $\submodel$-minimum $\tupof{\setof s,\emptyset,\emptyset,\emptyset}$ that contains no information whatsoever. A necessary condition for being a witness would therefore have to check for validity, which is more than can be expected from a simple ``syntactic'' condition such as $\witness$.

In fact, there are even simpler examples of evidence that satisfy $\witness$ nor $\counter$.

\begin{example}\label{ex:not-necessary}
Let $M=\tupof{S,\setof s,R,\emptyset}$ be a consistent model over some $F$ that includes $\phi=\EX\True$. Clearly, neither $\witness(s,\phi)$ nor $\counter(s,\phi)$ is satisfied; however, $M$ is either evidence for $(s,\phi,\ttt)$ (if $\succ(s)$ is non-empty) or for $(s,\phi,\fff)$ (if $\succ(s)$ is empty).\qed
\end{example}
One way to understand this example is that, although the model $M$ does not label the subformula $\True$ to have the value $L(s,\True)=\ttt$, it is not possible to label it any other way and still remain consistent. In other words, the labelling $L$ is \emph{constrained}, for some state/formula pairs where it is currently undefined, to assign a particular value in order to give rise to a consistent model. We will see that, in cases where the children of a formula are \emph{not} constrained in this way, $\witness$ and $\counter$ are necessary conditions for witnesses, respectively counterexamples. The following makes the notion of constrained formulas precise.
\begin{definition}[constrained formula set]\label{def:constrained}
Let $F\in \CForm$. A set of formulas $G\subseteq F$ is \emph{constrained} if there exists an inconsistent model $M\in\Model(F)$ with $\dom L\subseteq S\times G\pto \Bool$.
\end{definition}
\begin{example}\label{ex:constrained}
Here are some examples of constrained formula sets.
\begin{itemize}[smallsep]
\item $G=\setof{p,p\land q}$ is constrained, since any model with $s\in S$, $L(s,p)=\fff$ and $L(s,p\land q)=\ttt$ is inconsistent. Note that no proper subset of $G$ is constrained.
\item $\setof\True$ is constrained, since any model with $s\in S$ and $L(s,\True)=\fff$ is inconsistent.
\item $\setof{\EX\True}$ is constrained, since for any model $M$ with $s\in C$ and $L=\emptyset$, either $M\sqcup\tupof{\setof{(s,\EX\True) \mapsto\ttt}}$ or $M\sqcup\tupof{\setof{(s,\EX\True) \mapsto\fff}}$ is inconsistent (see also \Cref{ex:not-necessary}).\qed
\end{itemize}
\end{example}
On the other hand, the following proposition characterises a class of unconstrained formula sets.
\begin{restatable}[unconstrained formula sets]{proposition}{propunconstrained}
Let $F\in \CForm$ such that the only operators in $F$ are $\lneg$, $\lor$ and $\EU$, and let $G\subseteq F$ be non-empty. If all propositions occurring in $G$ are distinct, then $G$ is unconstrained.
\end{restatable}
We can now state the announced necessity of $\witness$ and $\counter$ for formulas of which the children are unconstrained.

\begin{restatable}[$\witness$ and $\counter$ are sometimes necessary]{theorem}{thmnecessary}\label{thm:evidence-only-if}
Let $F\in \CForm$, let $M\in \Model(F)$  and let $\phi\in F_\Oper$ such that $\child(\phi)$ is unconstrained and $\dom L\subseteq S\times\child(\phi)$. For any $s\in S$:
\begin{enumerate}[smallsep]
\item If $M$ is a witness for $(s,\phi)$, then $\witness(s,\phi)$ holds;
\item If $M$ is a counterexample for $(s,\phi)$, then $\counter(s,\phi)$ holds.
\end{enumerate}
\end{restatable}
As noted before (e.g., \Cref{prop:evidence}), evidence is easy to come by, and may very well contain states and assertions that are not actually required for the formula to hold. We are in fact interested in \emph{simple} or \emph{small} evidence. In the remainder of this section, we consider \emph{minimal evidence} with respect to $\submodel$. \Cref{tab:min-evidence} gives conditions on models that we will show to characterise such minimal evidence.
\begin{table}
\caption{Minimal witness and counterexample conditions $\minwitness,\mincounter\of S\times F_\Oper\to\Bool$ for a given model $M\in\Model(F)$.}
\vspace*{-3mm}
\label{tab:min-evidence}
\[\begin{array}{l@{{\enspace}}l}
\phi & \minwitness_M(s,\phi) \\
\hline
\True
  & M=\tupof{\setof{s},\emptyset,\emptyset,\emptyset} \\
\lneg\psi
  & M=\tupof{\setof{s}, \emptyset, \emptyset, \setof{(s,\psi)\mapsto\fff}} \\
\psi_1\lor\psi_2
  & M=\tupof{\setof{s}, \emptyset, \emptyset, \setof{(s,\psi_i)\mapsto\ttt}} \text{ for } i=1 \text{ or } i=2 \\
\EX\psi
  & M=\tupof{\setof{s,s'}, \emptyset, \setof{(s,s')}, \setof{(s',\psi)\mapsto \ttt}} \\
\EUof{\psi_1}{\psi_2}
  & M=\tupof{[\rho],\emptyset,<_\rho,\gensetof{(s',\psi_1)\mapsto\ttt}{s'\in[\rho]\setminus \setof{\rho\last}} \cup \setof{(\rho\last,\psi_2)\mapsto \ttt}} \\
  & \qquad\text{where } \rho \in s\cc \State\finseq \text{ is acyclic}
   \\
\EG{\psi}
  & M=\tupof{[\rho],\max <_\rho,<_\rho,\gensetof{(s',\psi)\mapsto\ttt}{s'\in [\rho]}} \\
  & \qquad\text{where } \rho\in s\cc \State\finfseq \text{ such that every } \rho'\prf\rho \text{ is acyclic}
\\[\bigskipamount]
\phi & \mincounter_M(s,\phi) \\
\hline
\True
  & \text{never} \\
\lneg\psi
  & M=\tupof{\setof{s}, \emptyset, \emptyset, \setof{(s,\psi)\mapsto\ttt}} \\
\psi_1\lor\psi_2
  & M=\tupof{\setof{s}, \emptyset, \emptyset, \setof{(s,\psi_1)\mapsto\fff,(s,\psi_)\mapsto\fff}} \\
\EX\psi
  & M=\tupof{\setof s\cup \gensetof{s_i}{i\in I}, \setof s, \gensetof{(s,s_i)}{i\in I}, \gensetof{(s_i,\psi)\mapsto \fff}{i\in I}} \\
\EUof{\psi_1}{\psi_2}
  & M=\tupof{S,C,R,\gensetof{(s',\psi_1)\mapsto \fff}{s'\in S}\cup \gensetof{(s',\psi_2)\mapsto \fff}{s'\in S\setminus C}} \\
  & \qquad\text{where all of } S \text{ is $R$-reachable from } s \text{ and } C\supseteq S\setminus \max R \\
\EG{\psi}
  & M=\tupof{S,S\setminus\max R,R,\gensetof{(s',\psi)\mapsto \fff}{s'\in \max_R S}} \\
  & \qquad\text{where } S \text{ is finite, } R \text{ is acyclic and } \min R=\setof s.
\end{array}
\]
\end{table}
To give an intuition for these technical conditions, \Cref{fig:core-minimal} graphically represents some typical minimal witnesses and counterexamples, for top-level temporal operators. We briefly discuss them.
\begin{itemize}
\item For a witness for $\EX\psi$, it suffices to have one outgoing transition to a state satisfying $\psi$. That state may either be another state, or $s$ itself. All states can be open. Conversely, for a counterexample, \emph{all} successor states (which may or may not include $s$ itself) must \emph{fail} to satisfy $\psi$; and in addition, $s$ must be closed, otherwise there are supermodels in which $s$ has additional successors that do satisfy $\psi$.

\item A witness for $\EUof{\psi_1}{\psi_2}$ only requires a single (finite) path only, consisting of open states, the last of which satisfies $\psi_2$ whereas all other ones satisfy $\psi_1$.

\item To witness $\EG\psi$, again we need only a single path, but this may be infinite; and if not, it must end in a closed state (which guarantees that no supermodel can have a path strictly extending this one). All other states can remain open. Infinite paths may arise in two different ways: through a loop, or through an infinite sequence of distinct states. All states in the witness must satisfy $\psi$.
\end{itemize}
\begin{figure}
\centering
\subcaptionbox{Minimal $\EX$-witnesses\label{fig:EX-minimal-witness}}[.33\textwidth]%
  {\includegraphics[scale=.4]{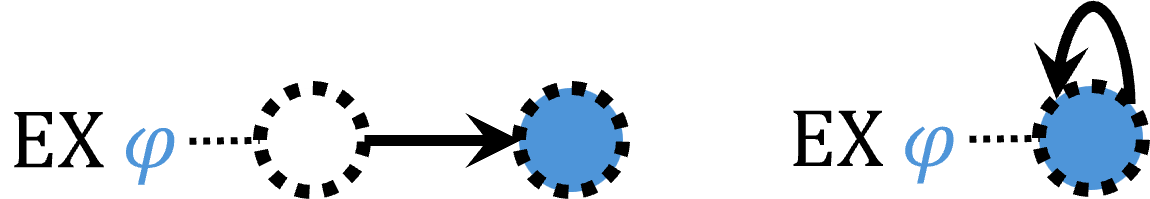}}%
%  \quad
\subcaptionbox{Minimal $\EU$-witness\label{fig:EU-minimal-witness}}[.33\textwidth]% 
  {\includegraphics[scale=.4]{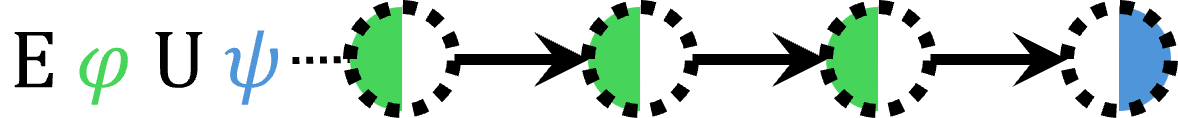}}%
%  \quad
\subcaptionbox{Minimal $\EG$-witnesses\label{fig:EG-minimal-witness}}[.33\textwidth]%
  {\includegraphics[scale=.4]{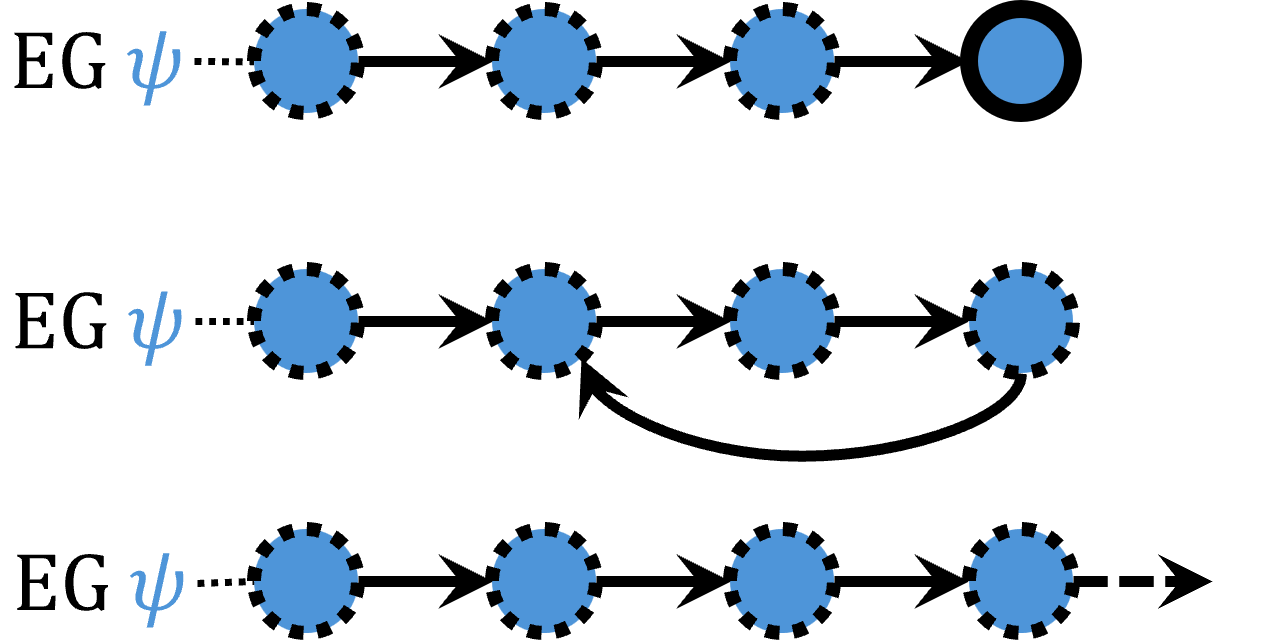}}%
  \\[\medskipamount]
\subcaptionbox{Minimal $\EX$-counterexamples\label{fig:EX-minimal-counterexample}}[.33\textwidth]%
  {\includegraphics[scale=.4]{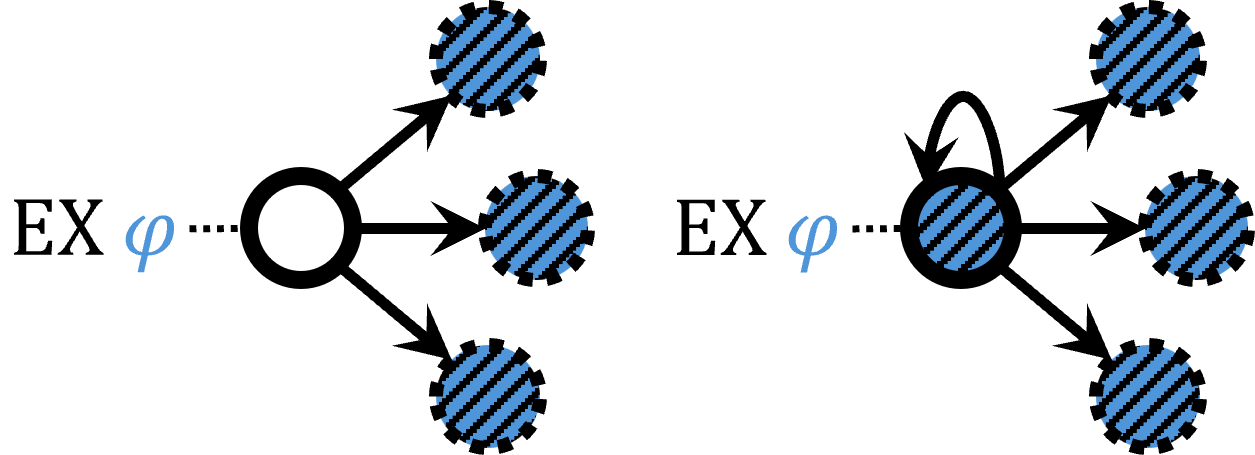}}%
%  \quad
\subcaptionbox{Minimal $\EU$-counterexample\label{fig:EU-minimal-counterexample}}[.33\textwidth]%
  {\includegraphics[scale=.4]{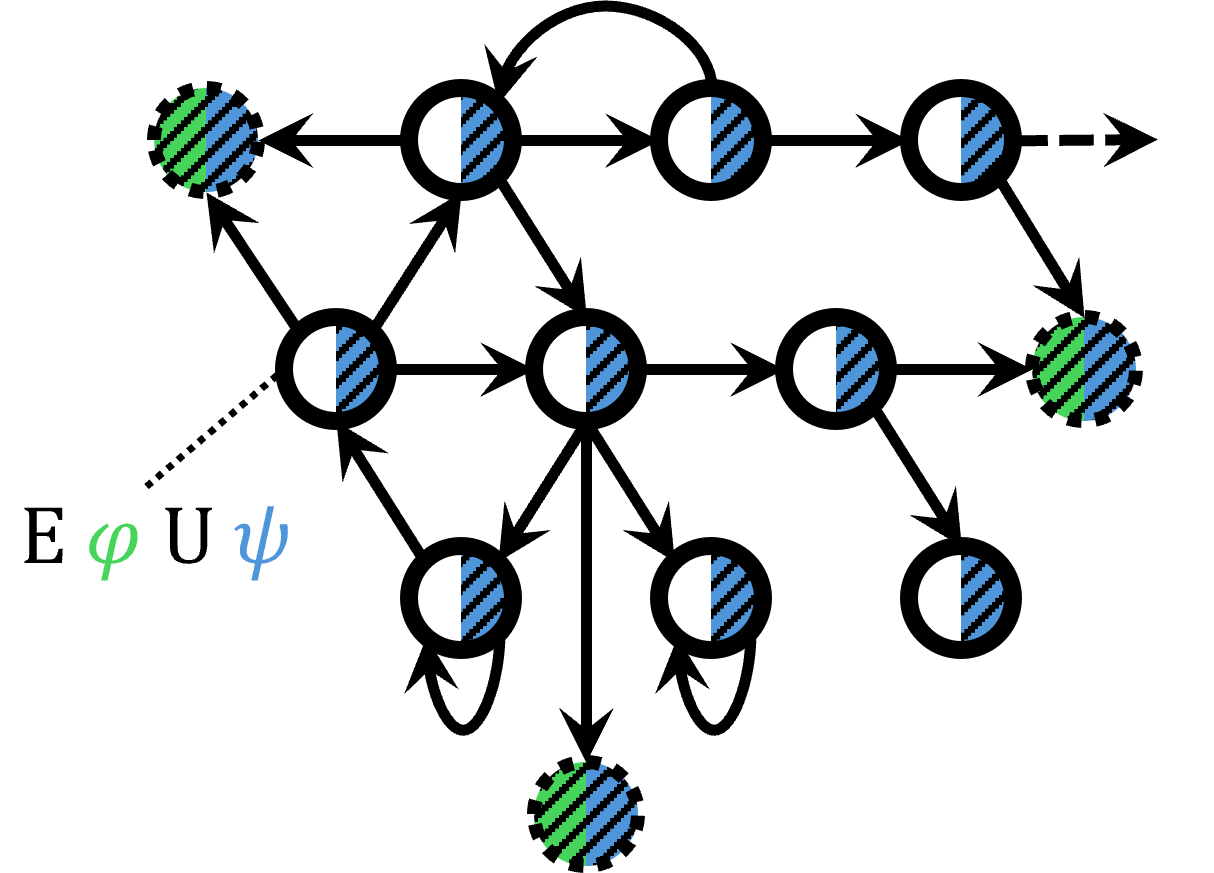}}%
%  \quad
\subcaptionbox{Minimal $\EG$-counterexample\label{fig:EG-minimal-counterexample}}[.33\textwidth]%
  {\includegraphics[scale=.4]{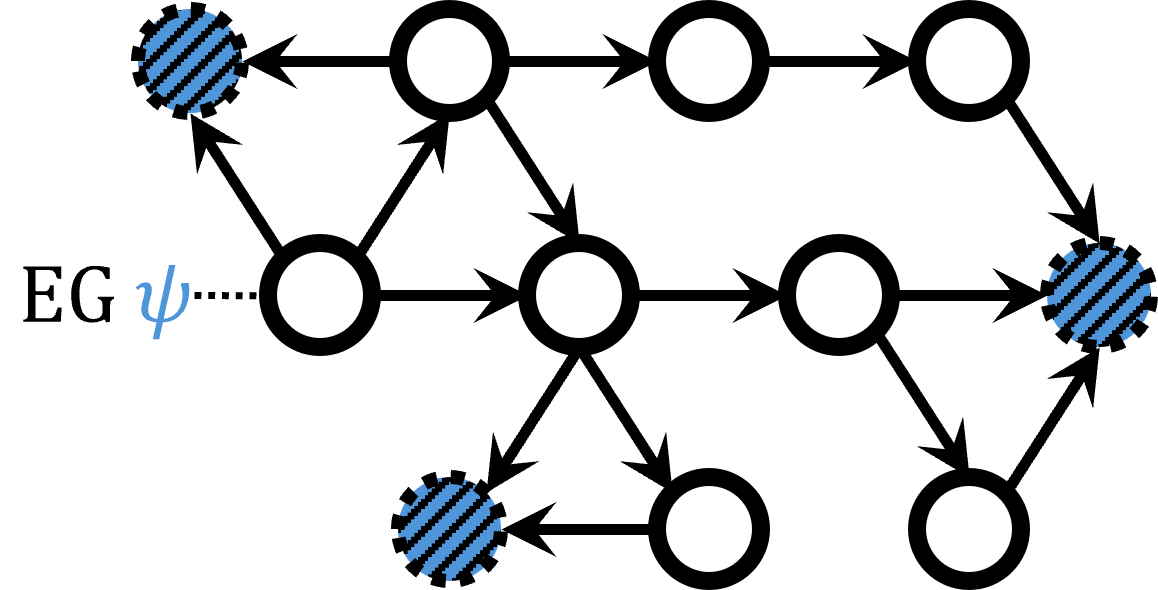}}%
\caption{Sample minimal evidence for temporal formulas. Dotted nodes are open; dashed hatch patterns stand for negated formulas; dashed arrows stand for infinite chains of states with the characteristics of their sources.}
\label{fig:core-minimal}
\end{figure}
The following result states that the conditions in \Cref{tab:min-evidence} indeed characterise $\submodel$-minima among the models that satisfy the conditions in \Cref{tab:evidence}.

\begin{restatable}[minimal syntactic evidence]{theorem}{thmminevidence}\label{thm:min-syntactic-evidence}
Let $F\in\CForm$, $M\in\Model(F)$, $s\in S$ and $\phi\in F_\Oper$.
\begin{enumerate}
\item $M$ is a $\submodel$-minimum of $\gensetof{N\in \Model(F)}{\witness_N(s,\phi)}$ iff $\minwitness_M(s,\phi)$ holds.
\item $M$ is a $\submodel$-minimum of $\gensetof{N\in \Model(F)}{\counter_N(s,\phi)}$ iff $\mincounter_M(s,\phi)$ holds.
\end{enumerate}
\end{restatable}
For unconstrained formulas, the \emph{semantic} counterpart of \Cref{thm:min-syntactic-evidence} holds.

\begin{corollary}[minimal semantic evidence]\label{co:minimality}
Let $F\in\CForm$, $M\in\Model(F)$, $s\in S$ and $\phi\in F_\Oper$ such that $\child(\phi)$ is unconstrained and $\dom L\subseteq S\times \child(\phi)$.
\begin{enumerate}
\item $M$ is a $\submodel$-minimal witness for $(s,\phi)$ iff $\minwitness_M(s,\phi)$ holds.
\item $M$ is a $\submodel$-minimal counterexample for $(s,\phi)$ iff $\mincounter_M(s,\phi)$ holds.
\end{enumerate}
\end{corollary}

\section{Proofs and their Visualisation}
\label{sec:proofs}

The purpose of evidence visualisation is to aid the comprehension of model checking outcomes. In terms of this paper, this translates to ``understanding why a given model is sound'' --- which in turn is a matter of providing evidence for all its compound assertions. We formalise this in the notion of \emph{proof}.

\begin{definition}[proof]\label{def:proof}
Let $F\in \CForm$. A \emph{proof} over $F$ is a tuple $P=\tupof{M,\pi}$ such that $M\in\Model(F)$ and $\pi\of \hat L_\Oper\mapsto \gensetof{N\in \Model(F)}{N\submodel M}$ such that, for all $a\in \hat L_\Oper$, $\pi(a)$ is evidence for $a$.
\end{definition}
The following proposition makes clear that proofs are a tool to establish consistency. Recalling that sound models are just consistent full ones, as a special case we get that proofs are a tool to establish soundness.

\begin{restatable}{proposition}{propproof}\label{prop:proof}
Let $F\in \CForm$. A model $M\in\Model(F)$ is consistent if and only if $M\submodel N$ for some proof $\tupof{N,\pi}$ over $F$.
\end{restatable}
\begin{note}[complexity]\label{note:complexity}
Although in this paper we are not concerned with questions of complexity, it is worth pointing out that proofs can be constructed, at no extra cost, as a byproduct of (explicit-state) \CTL model checking. Very roughly, to understand why this is the case, recall that checking a compound formula $\psi=(o,\bar\phi)$ is a recursive procedure, in which first all $\bar\phi\proj i$ are checked, followed by a state space traversal depending on the operator $o$, in which all states are marked either $\ttt$ or $\fff$ for $\psi$. Generating a proof (in the sense of \cref{def:proof}) is a matter of recording the reason why a $\ttt$- or $\fff$-mark is given.

For instance, to check $\psi=\EUof{\phi_1}{\phi_2}$, the standard algorithm immediately marks all states that satisfy $\phi_2$ with $\ttt$, after which (recursively) the incoming transitions of all freshly $\ttt$-marked states are explored backwards, and their sources are also marked $\ttt$ if they additionally satisfy $\phi_1$. After this stage has terminated, all remaining states are marked $\fff$. The proof generation starts with all states of the system, none of which are closed, and no transitions. For the positive (witness) part of the proof, all transitions that give rise to a $\ttt$-mark of their source should be added. Witnesses are minimal if the backward exploration is done breadth-first. For the negative (counterexample) part of the proof, every $\fff$-marked state that satisfies $\phi_1$ should be closed and all its outgoing transitions should be added. The labelling of the proof is obvious.\qed
\end{note}
Minimal evidence for any individual assertion can be visualised in a straightforward manner as a transition system, using visual cues for open states and three-valued colour coding for the labelling (e.g., green for $\ttt$, red for $\fff$ and grey for undefinedness). The visualisation of a complete proof is obviously more complicated. One solution is to do this interactively, by showing the evidence upon selection of a state/formula pair.

Our visualisation of model checking outcomes, evidence and proof is based on the abstract syntax tree (AST) of the formulas, shown in every state. Every AST node is coloured according to the satisfaction of the subtree depending on it in the state concerned. Hence, for instance, the formula $\EUof{\,\lneg\done\,}{\win}$ (expressing that the game can be won without ever throwing~1), when checked on this system, gives rise to the full model shown in \Cref{fig:E-win-complete} --- which shows that the formula is actually only satisfied in state $s_{12}$. (All visualisation snapshots in this section have been produced using the tool demonstrator \cite{CTLviz-artefact}.)

\begin{figure}[t]\centering
\subcaptionbox{Model checking outcome\label{fig:E-win-complete}}[.49\textwidth]%
  {\includegraphics[scale=.2]{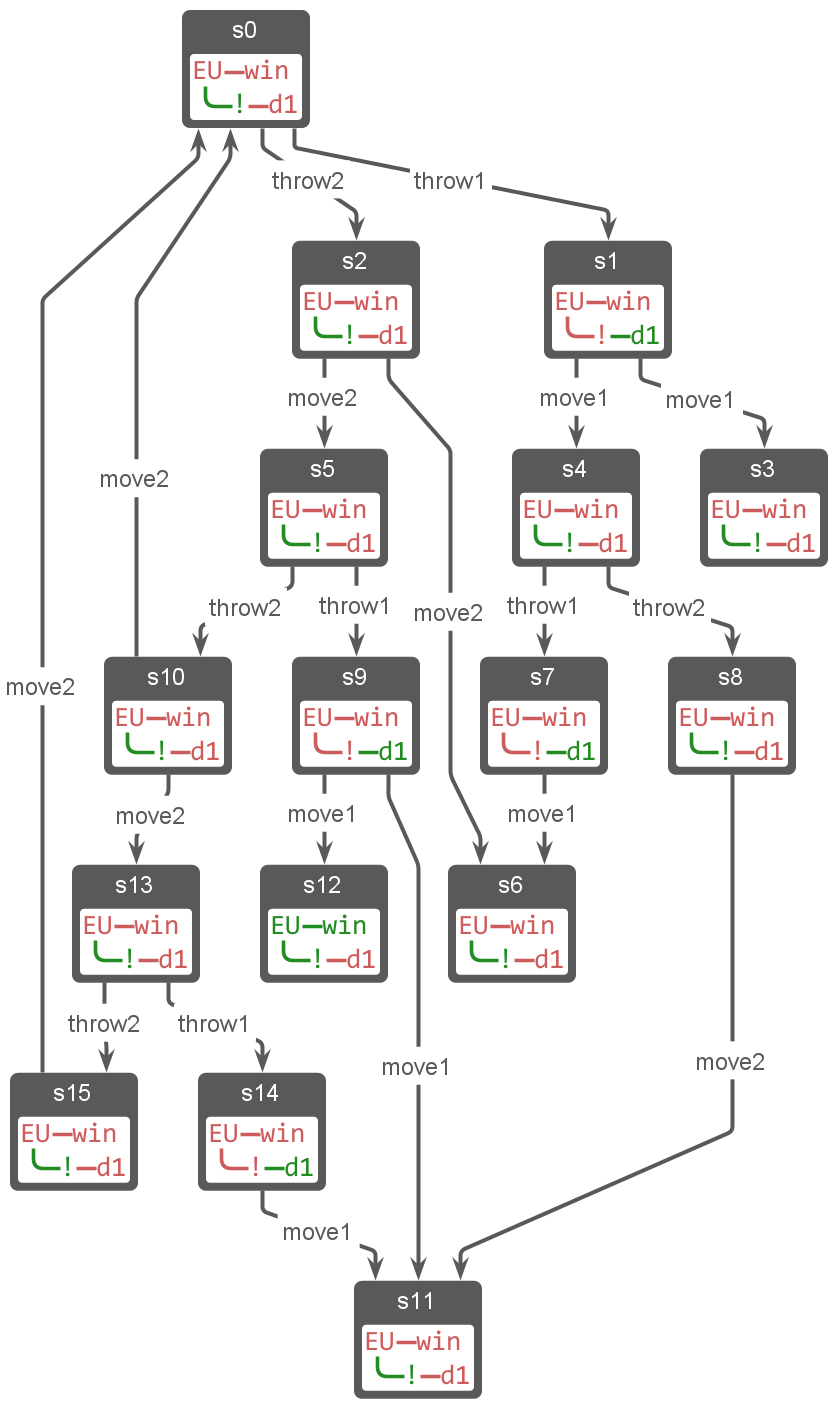}}
\subcaptionbox{Minimal evidence. Dotted states are open; states and formula nodes outside the witness are greyed out\label{fig:E-win-minimal}}[.49\textwidth]%
  {\includegraphics[scale=.2]{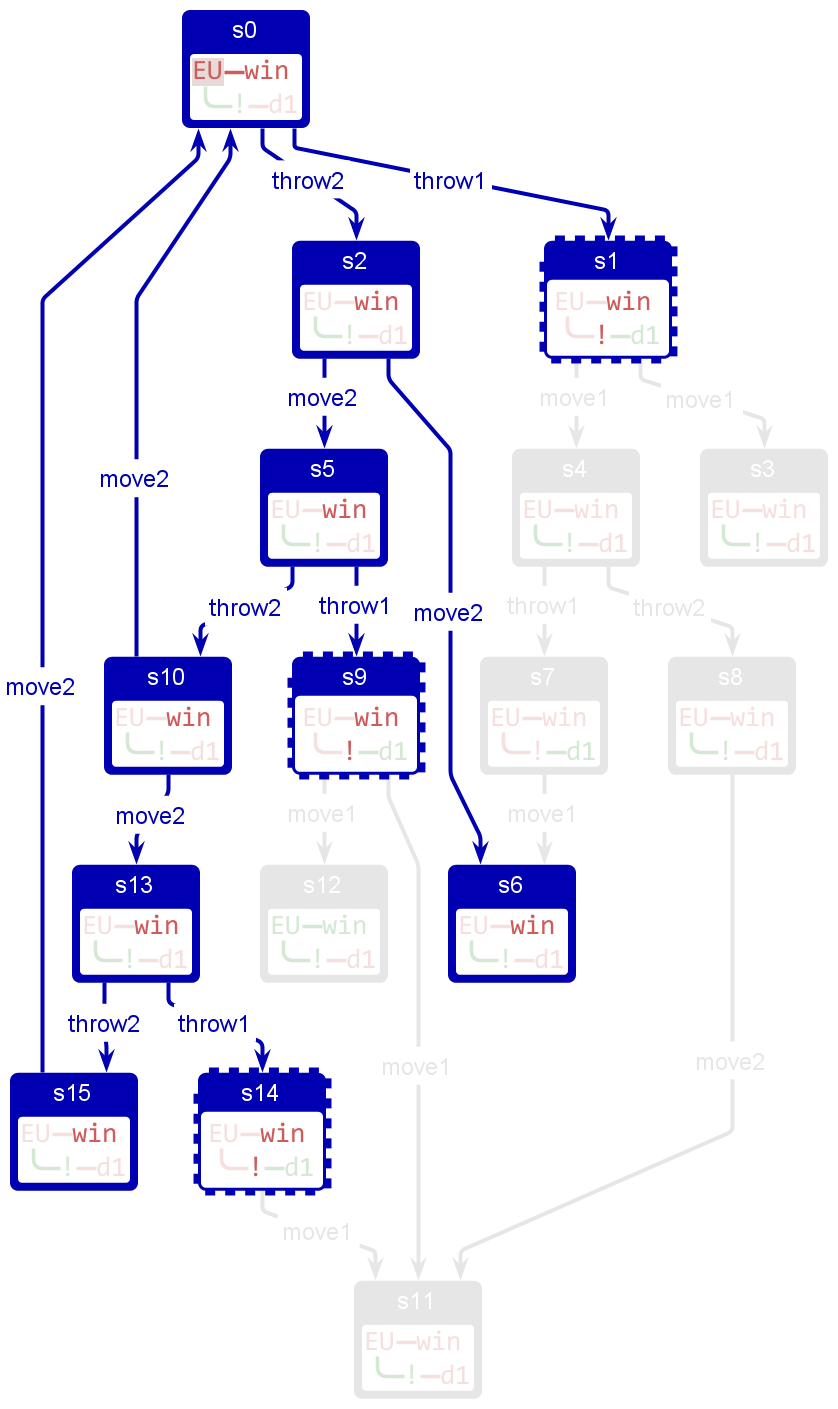}}
\caption{Model checking result and minimal witness for $\EUof{\,\lneg\done\,}{\win}$. The ASTs are displayed with roots top-left; children are ordered from bottom to top}
\label{fig:E-win}
\end{figure}
\Cref{fig:E-win-minimal} shows the minimal evidence on the initial state. Since the formula is not satisfied, the evidence is a counterexample, which is an instance of the template shown in \Cref{fig:EU-minimal-counterexample}. Note that, indeed, the blue part contains all the information required to conclude that $\EUof{\,\lneg\done\,}{\win}$ is not satisfied on $s_0$: in particular, $\win$ is not satisfied in any of the states shown, and the maximal paths either contain closed states only (either ending in the terminal state $s_6$ or cycling through $s_{10}$ and possibly $s_{15}$), or end in an open state where $\lneg\done$ is not satisfied ($s_1$, $s_9$ and $s_{14}$).

From this first visualisation of evidence, recalling that the main purpose is human comprehension, there are two immediate improvements that can be made. Both raise the amount of relevant information in the figure.
\paragraph{Local closure.}

The satisfaction of local (i.e., non-temporal) operators only depends on the satisfaction of their children in the same state. It can therefore not give rise to confusion in include the evidence for those children in the same picture. The result is shown in \Cref{fig:E-win-minimal-local}. (For this example, the difference is small: the only local operator is the negation.)

\paragraph{Natural evidence.}

\newcommand{\greenphi}{\textcolor{greenformula}{\phi}}%
\newcommand{\bluepsi}{\textcolor{blueformula}{\psi}}%
Minimal witnesses for $\EU$ and minimal counterexamples for either $\EU$ and $\EG$ do not have full information about the satisfaction of all children for any of the non-terminal states (see \Cref{fig:core-minimal}: all non-terminal states are blank for one of the children). Though, as we have seen (\Cref{co:minimality}), such information is indeed formally superfluous, we feel that omitting it actually hampers comprehensibility. For instance, consider the minimal witness shown in \Cref{fig:EU-minimal-witness}. If the path of that witness is $\tupof{s_0\,s_2\,s_3\,s_4}$ then it is a submodel of all of the following models, hence evidence that $\EUof\greenphi\bluepsi$ holds for $s_0$:
\[ \includegraphics[scale=.4]{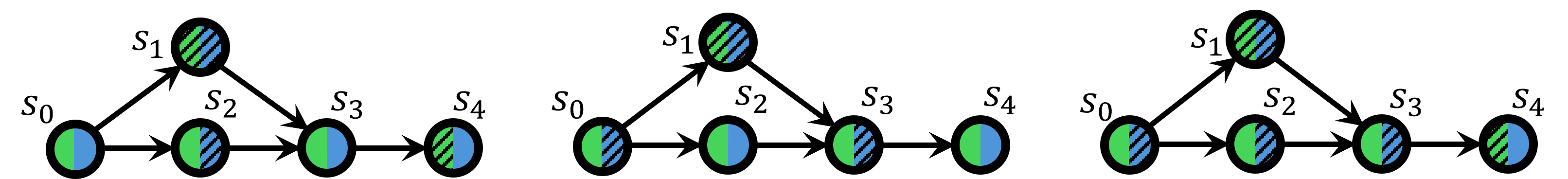} \]
In the first two of these models, this particular minimal witness is \emph{not} the most obvious explanation of $\EUof\greenphi\bluepsi$: both have ``shorter'' witnesses, in terms of state count, consisting of just state $s_0$ with $L=\setof{(s_0,\bluepsi)\mapsto\ttt}$, respectively $\tupof{s_0\,s_2}$ with $L=\setof{(s_0,\greenphi)\mapsto\ttt, (s_2,\bluepsi)\mapsto\ttt}$. To capture the distinction between the non-obvious, overlong witnesses and the more intuitive, shorter ones, we introduce the concept of \emph{natural} evidence.

\begin{definition}[natural evidence]
Evidence $M$ for an assertion $a\in S\times F_{\EU,\EG}\times\Bool$ is \emph{natural} if $L(s,\psi)$ is defined for all $s\notin \max_R S$ and $\psi\in \child(\phi_a)$.
\end{definition}
Hence $\EU$- and $\EG$-evidence is natural when its labelling includes the satisfaction of all children in all non-terminal states. We believe that minimal \emph{natural} evidence is more informative than minimal evidence. Rather than Figs.\ \ref{fig:EU-minimal-witness}, \ref{fig:EU-minimal-counterexample} and~\ref{fig:EG-minimal-counterexample}, typical examples of minimal natural evidence are
\begin{center}
\begin{tabular}{c@{{\enspace}}c@{{\enspace}}c}
\includegraphics[scale=.4]{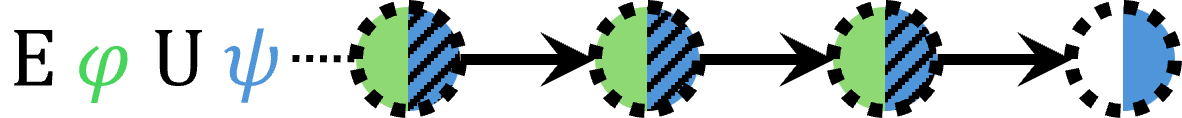}
& \includegraphics[scale=.4]{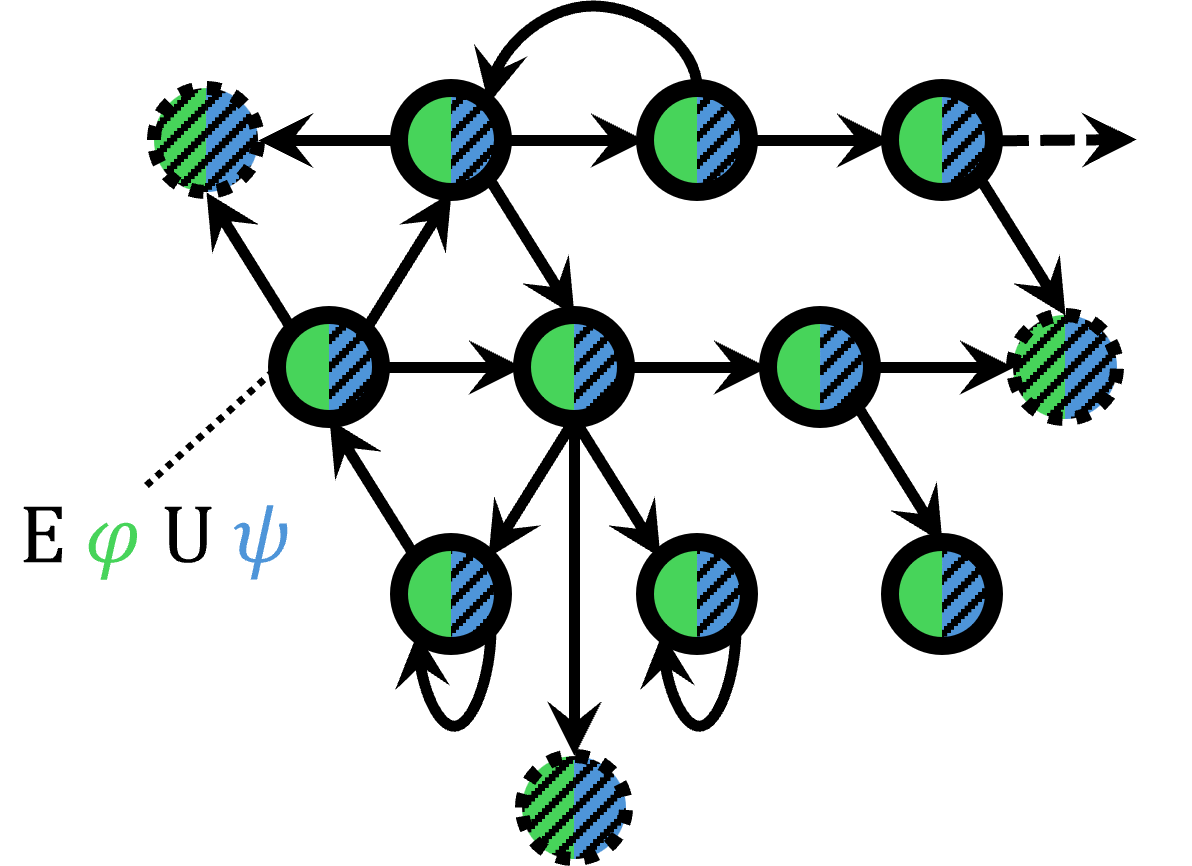}
& \includegraphics[scale=.4]{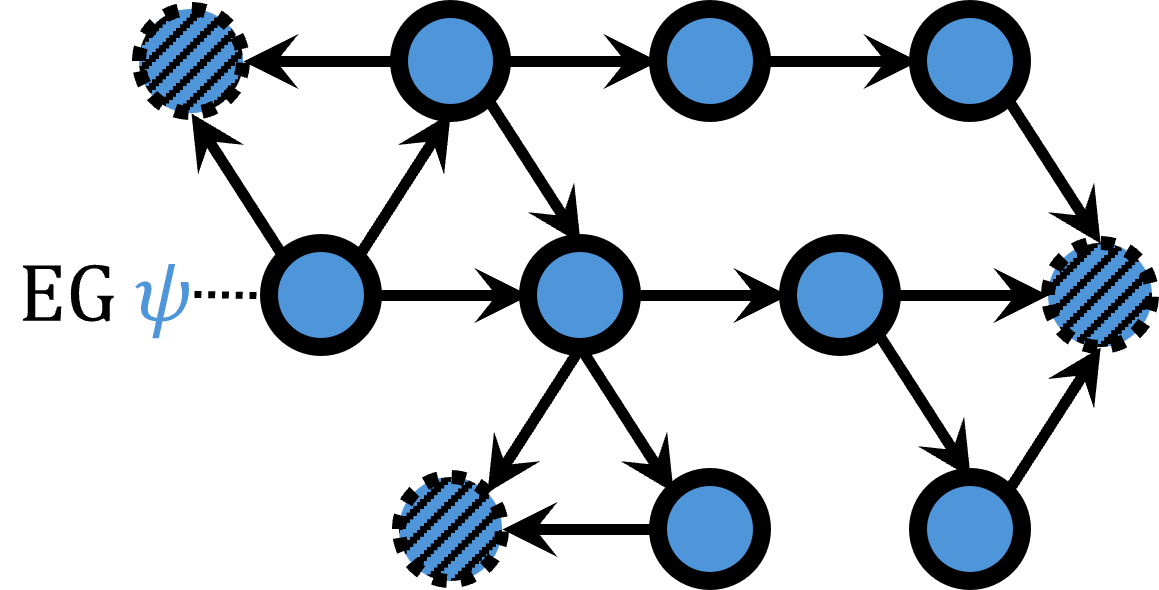}
\\
\tiny Natural $\EU$-witness
& \tiny Natural $\EU$-counterexample
& \tiny Natural $\EG$-counterexample
\end{tabular}
\end{center}
For the example formula of \Cref{fig:E-win}, the minimal natural counterexample is shown in \Cref{fig:E-win-minimal-natural-local}. Here it is clear that, in all intermediate states, $\lneg\done$ holds, and hence any counterexample must include all states shown here.

\begin{figure}[t]\centering
\subcaptionbox{Locally closed minimal evidence\label{fig:E-win-minimal-local}}[.5\textwidth]%
  {\includegraphics[scale=.2,trim={0cm 6cm 5cm 0},clip]{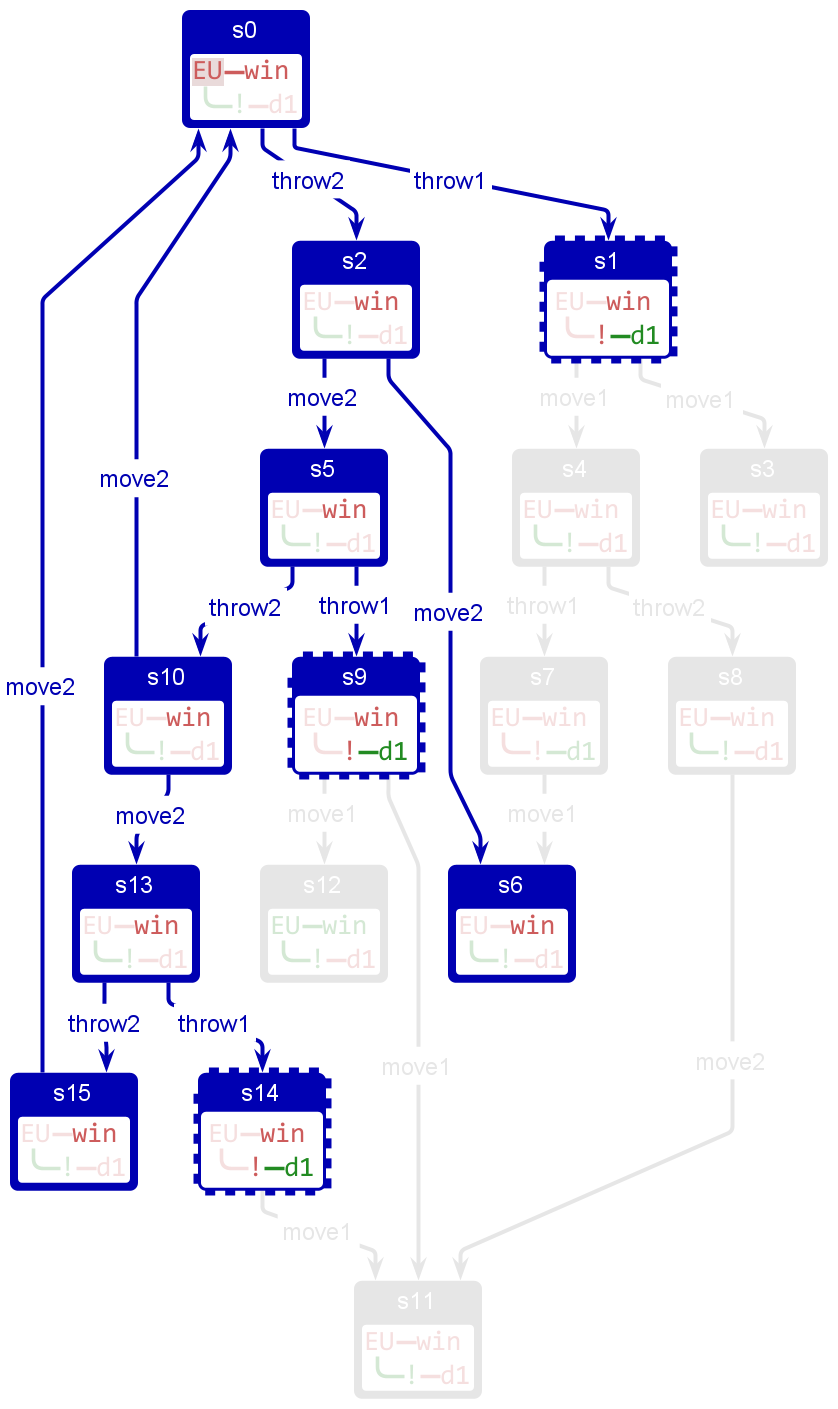}}%
\subcaptionbox{Locally closed minimal natural evidence\label{fig:E-win-minimal-natural-local}}[.5\textwidth]%
  {\includegraphics[scale=.2,trim={0cm 6cm 5cm 0},clip]{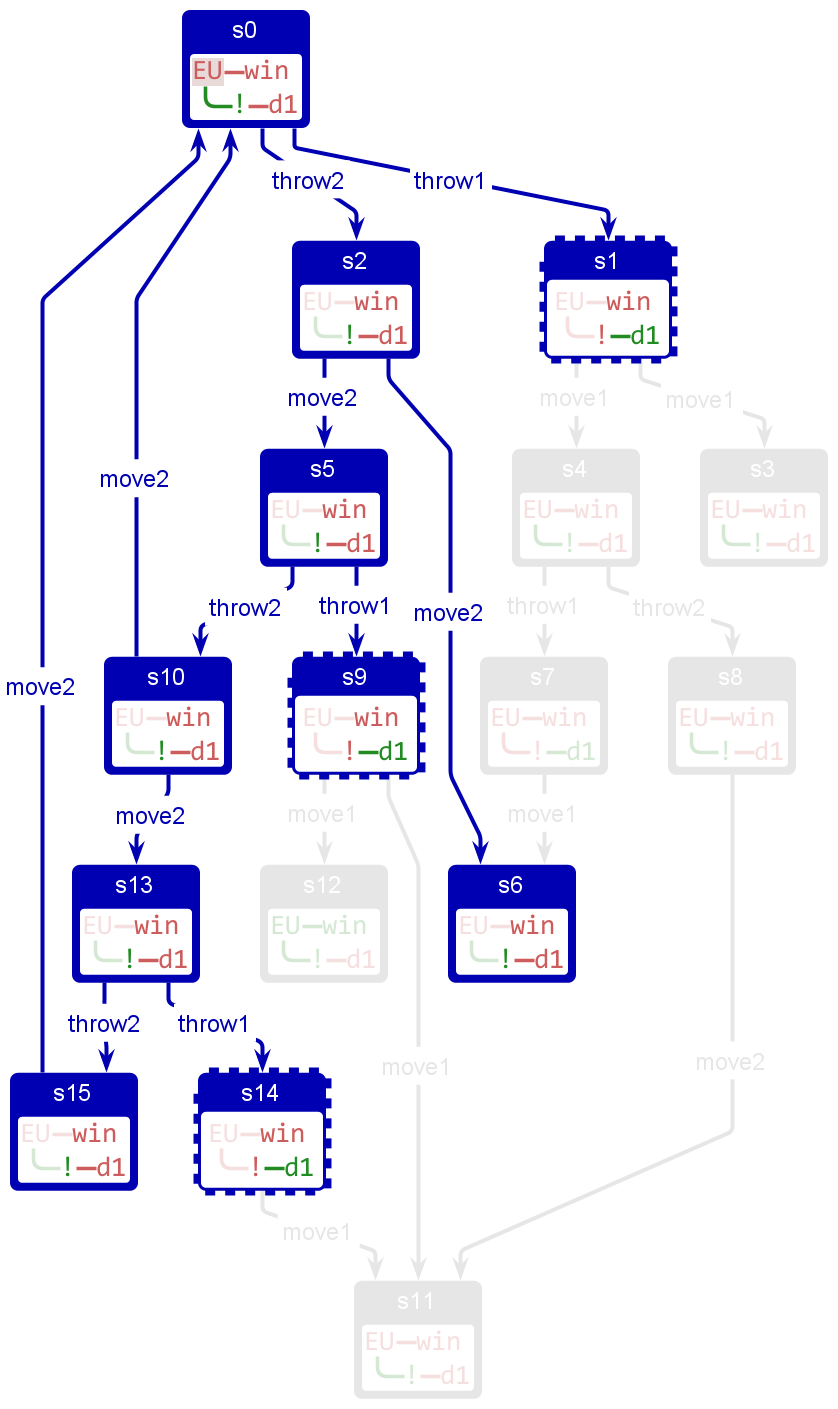}}
\caption{Locally closed and natural evidence for $\EUof{\,\lneg\done\,}\win$ (compare \cref{fig:E-win-minimal}).}
\end{figure}

\section{Combined Evidence}
\label{sec:combined}

The visualisation of a \CTL proof typically involves many individual snapshots, each one showing the evidence for a given pair of state and temporal subformula. Even if these can be called up by interactively selecting state/formula instances, it is still hard to obtain a comprehensive understanding.

Combining different evidence models into a single static representation, on the other hand, generally runs into the problem that (in particular for path-based temporal formulas) each individual witness and counterexample is a model of potentially unbounded size, which for different states will generally overlap. The union of two evidence models may very well fail to be evidence, hence its visualisation provides no useful information.

Nevertheless, somewhat surprisingly, a specific form of evidence combination is actually possible: namely, for a given model and a single, fixed formula.

\begin{definition}[combined evidence]\label{def:combined}
Let $F\in\CForm$ and $\phi\in F$, and let $K\of S\times \setof \phi\to\Bool$ be a total labelling. A model $M\in\Model(F)$ is \emph{combined [natural] evidence for $\phi$} if it is [natural] evidence for all $a\in \hat K$.
\end{definition}
In the next proposition, we use $M\proj s$ (for $s\in S$) to denote the restriction of $M$ to the states reachable from $s$.

\begin{restatable}[combined evidence always exists]{proposition}{propcombinedevidenceexists}\label{prop:combined}
Let $F\in\CForm$ and let $M\in\Model(F)$ be sound. For every $\phi\in F_\Oper$ there exists combined evidence $E_\phi\submodel M$. If, moreover,  $\phi\in \setof{\EU,\EG}\times \Form\finseq$, then there exists combined [natural] evidence $E_\phi\submodel M$ such that for all $s\in S$, $E_\phi\proj s$ is minimal [natural] evidence.
\end{restatable}

\begin{figure}[t]
\subcaptionbox{Evidence for the top-level $\EG$, selected for $s_0$}[.5\textwidth]%
  {\includegraphics[width=.5\textwidth]{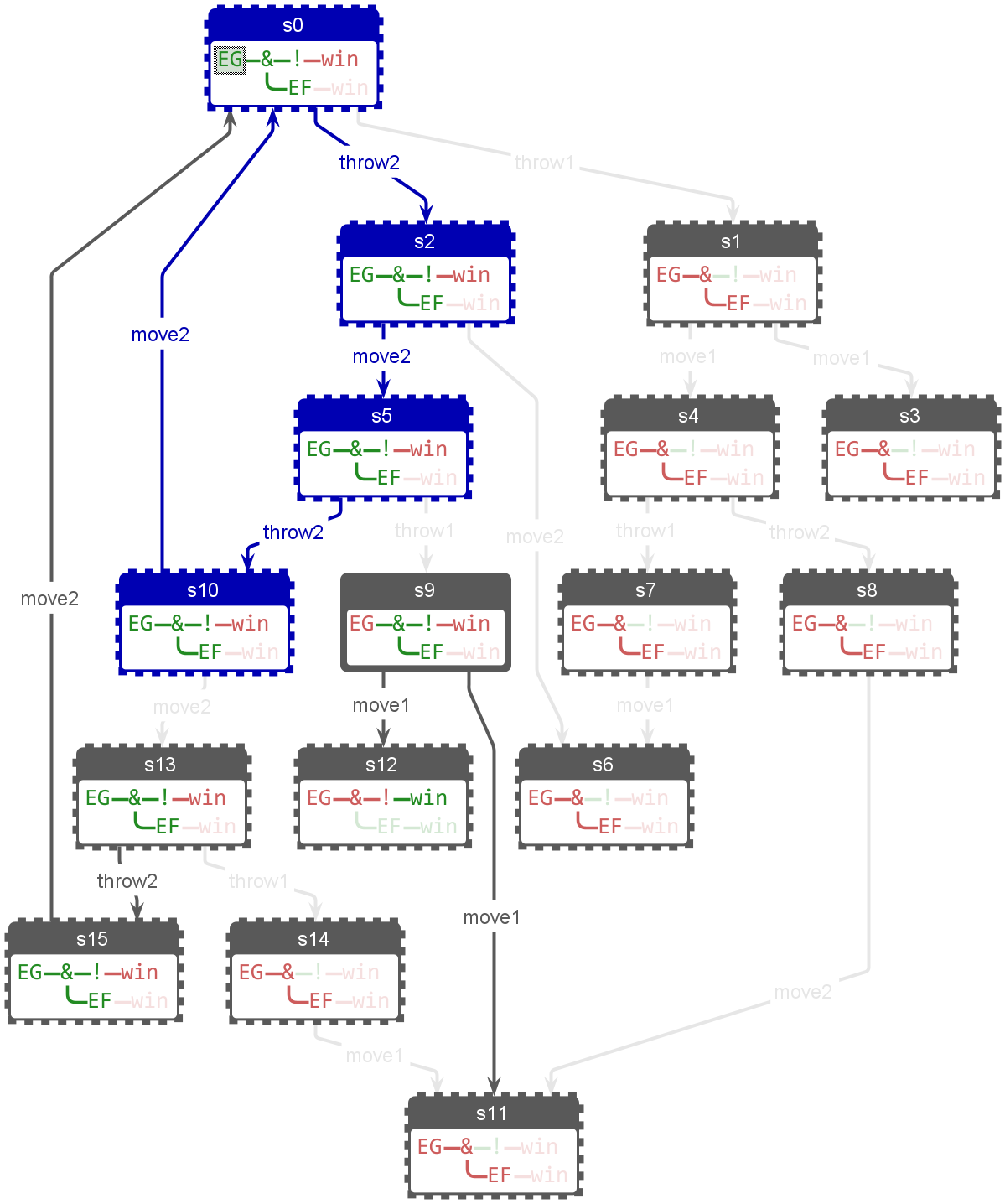}}%
\subcaptionbox{Evidence for the nested $\EF$, selected for $s_{5}$}[.5\textwidth]%
  {\includegraphics[width=.5\textwidth]{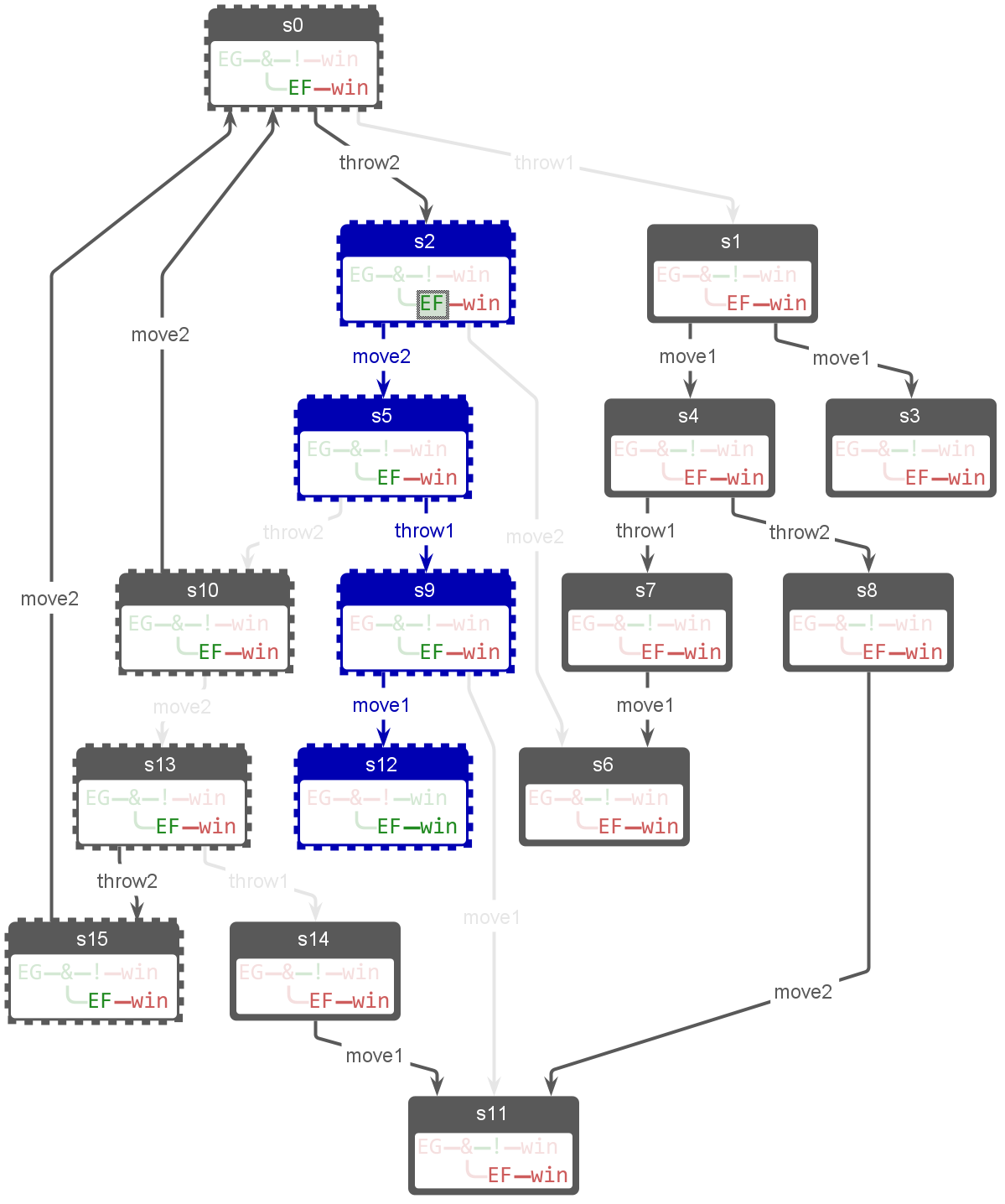}}
\caption{Combined natural evidence for the temporal operators of $\EG (\lneg\win\land \EF \win)$.}
\label{fig:combined}
\end{figure}

\begin{example}\label{ex:combined}
Consider once more $\phi=\EG (\lneg\win\land \EF \win)$ (cf.\ \Cref{ex:game}). \Cref{fig:combined} shows combined locally closed, natural evidence for the two temporal operators of $\phi$ based on the model of \Cref{fig:state-space}. The fragments reachable from the selected states $s_0$ resp.\ $s_5$ have been coloured blue; however, these two models together actually explain the entire model checking outcome for all states.\qed
\end{example}
\section{Related work}
\label{sec:related}

There is substantial prior work on the notion of evidence, some of which is actually in the context of (much) more powerful logics, such as \CTL enriched with description logic \cite{ALCCTL}, the modal $\mu$-calculus \cite{Efficient-Evidence-Generation} or even more general fixpoint logics \cite{Evidence-for-Fixpoint-Logic}. Comprehensive overviews of diagnostics for model checking can be found in \cite{Rich-Explanations,Review-on-Counterexample-Explanation,Counterexamples-Survey}.

To start with, \emph{linear} evidence (in our setting: evidence built on $\tupof\rho$ for some $\rho\in\State\finfseq$) has been studied in \cite{Counterexamples-and-Witnesses}, where it was shown how such evidence can be generated efficiently (using symbolic model checking) for fragments of \CTL and also \CTLstar for which this type of evidence actually exists. In the same vein, \cite{ACTL-Linear-Counterexamples} studies the complexity of determining whether linear evidence (in this case always a counterexample) exists for a given \ACTL formula (the universal fragment of \CTL). Linear evidence was extended in \cite{Tree-Like-Counterexamples} to \emph{tree-like} evidence for \ACTL and certain $\omega$-regular extensions. In \cite{Witness-Generation-Existential-CTL} it is shown how to efficiently generate \emph{minimal} evidence, in this case for the dual existential fragment $\ECTL$. (It should be pointed out that, in all the cited papers, the terms ``evidence'', ``witness'' and ``counterexample'' do not refer to a single operator, as in this paper, but to the entire nested structure that we have called ``proof''. Moreover, \emph{minimality} is defined only relative to the set of states, and not to the labelling.)

In contrast, in this paper we have not concerned ourselves with the complexity of \emph{generating} evidence (except for the brief excursion in \cref{note:complexity}), but with its (as precise as possible) \emph{characterisation}. The essential technical novelty are the closed states, which enable us to generalise to full CTL. For the fragments studied in the papers above, our evidence essentially reduces to their linear and tree-like counterpart, in which all states are open.

One further point of difference is that \cite{Tree-Like-Counterexamples} is based on \emph{simulation} of evidence by the Kripke structures to which it applies, where we consider \emph{supermodels}. In principle, simulation is the more general concept, as it abstracts from state identities. Our motivation to concentrate on sub- and supermodels, rather than going for a simulation-like relation, lies in the desire to obtain a human-comprehensible explanation for model checking outcomes: the fact that the evidence is displayed on the same states as the model that has been checked, rather than having to be understood through the indirection of a simulation relation, is actually an advantage in this context. In any case, the standard notion of simulation does not work as required for our concept of evidence, since it does not take closed states into account. For instance, consider once more the counterexample (say $M_1$) of \cref{fig:E-win-minimal} (ignoring the greyed-out part). Now imagine a model $M_2$ that extends $M_1$ with a single transition from $s_0$ to a new state $s_{15}$ in which \win is satisfied. We would have that $M_2$ simulates $M_1$, while $M_2$ does satisfy the property $\EUof{\,\lneg\done\,}\win$ to which $M_1$ is supposed to be a counterexample. We briefly return to this point in the future work discussion below.

In \cite{ALCCTL}, the logic is extended by replacing atomic predicates by description logic statements; however, the notion of evidence is essentially still restricted to tree-like structures. In particular, witnesses for properties of the form $\AF \phi$ are (essentially) just given as $\ttt$ or $\fff$, providing no useful structural information.

A very powerful notion of evidence, based on proof graphs for Parameterised Boolean Equation Systems, is studied in \cite{Evidence-for-Fixpoint-Logic,Evidence-Extraction-from-PBES,Efficient-Evidence-Generation} for logics that go far beyond \CTL in expressive power, up to the \emph{fixpoint logic} in \cite{Evidence-for-Fixpoint-Logic} that generalises the modal $\mu$-calculus. They show that such proof graphs can be generated effectively and, in some cases, efficiently. However, with great power comes greater complexity; the visualisation of general proof graphs has not been addressed systematically. With respect to proof graphs, our closed states can be thought of as concisely encoding a lot of \emph{negative} information --- namely, about the absence of transitions; as argued before, this is a contributing factor to the relative ease of representing evidence for path-existential operators. For instance, once more taking \cref{fig:E-win-minimal} as an example, an equivalent proof graph would have to contain a node for each of the closed states to explicitly represent the \emph{absence} of a transition from that state to the greyed-out state $s_{12}$, as adding such a transition would result in a model satisfying $\EUof{\,\lneg\done\,}\win$.

The work in \cite{Framework-for-Counterexample-Generation} is motivationally quite close to ours, driven by similar considerations of visualisation and comprehensibility of \CTL model checking outcomes. The visualisation framework presented there is in some ways more powerful than ours, since it involves interaction during the generation of evidence. The nature and visualisation of their counterexamples is close to the proof graphs cited in the previous paragraph; there are no considerations of minimality. Finally, \cite{Counterexamples-for-Classification} also pursues comprehensibility of model checking evidence, in this case for the modal $\mu$-calculus, but only for trace-like counterexamples.

\section{Conclusion and future work}
\label{sec:conclusion}

We have answered the questions posed in the introduction as follows.

\begin{enumerate}[label=\emph{(\roman*)},topsep=\smallskipamount]
\item \emph{What is the simplest notion of evidence that precisely captures the reason for satisfaction or violation of an arbitrary \CTL formula?}

We have given a formal characterisation of \emph{minimal evidence} for individual operators (\Cref{co:minimality}), leading to proofs for entire formulas (\Cref{def:proof}) that map every state/subformula pair to corresponding evidence. The main new concept here is that of \emph{closed states}, concisely capturing the negative information about the absence of further outgoing transitions.

\item \emph{How can we effectively visualise this notion of evidence?}

We have implemented an interactive visualisation demonstrator \cite{CTLviz-artefact}, available as software artefact accompanying this paper, which shows the evidence upon selecting any state/subformula pair (\cref{sec:proofs}). To make the visualisation more effective, we have limited the number of individual snapshots required to collectively represent a proof through \emph{local closure}, in which the evidence of non-temporal subformulas is added to that of temporal operators without risk of confusion or loss of information. Furthermore, we have introduced the concept of \emph{natural} evidence to increase comprehensibility and to rule out cases that, though minimal in a formal sense, are not intuitively obvious. Finally, we have shown that it is actually possible to combine the (natural) evidence for a given subformula over \emph{all} states in a single model (\Cref{sec:combined}), such that the minimal evidence for a given state emerges as the reachable fragment of that combined evidence. Thus, in the end, the number of evidence models that collectively represent a complete proof is reduced to one per temporal operator, rather than one per state/subformula pair.
\end{enumerate}
We see two concrete avenues for future research.

\paragraph{More expressive logics.}It is worth investigating to what extent the concepts of this paper (especially the particular notion of evidence) remain useful in more expressive logics. The primary candidate is $\CTLstar$. Obviously, the difficulty lies in the fact that the assertions are then not based on state/formula pairs, but on path/formula pairs. However, since all such formulas are still bound by a universal or existential path quantifier, it might be possible to combine the evidence for paths from a given state and so retain some of the advantages of the current setup.

\paragraph{Simulation in the presence of closed states.} In \Cref{sec:related} we have recalled that \cite{Tree-Like-Counterexamples} uses simulation to relate evidence (in that case tree-like counterexamples) and the Kripke structures to which that evidence applies: if a Kripke structure simulates a counterexample, the property is violated in the Kripke structure. For us, instead, evidence applies to any \emph{supermodel}. We have sketched why simulation not work ``out of the box'' in the setting of this paper: it allows the (simulating) model to have transitions from states that are closed in the (simulated) evidence. We hypothesise that our results can be generalised to a simulation-like relation by requiring \emph{bisimulation} on closed states.

\subsection*{Data Availability} The demonstrator tool for evidence visualisation accompanying this paper, and used for the screenshots of \cref{fig:state-space,fig:core-minimal,fig:combined}, is available as software artefact, and has received the ``Availability'' and ``Reusability'' badges. The artefact is archived and can be found at \cite{CTLviz-artefact}.

\bibliographystyle{splncs04}
\bibliography{refs}

@article{Counterexamples-Survey,
  author       = {Hichem Debbi},
  title        = {Counterexamples in Model Checking - {A} survey},
  journal      = {Informatica (Slovenia)},
  volume       = {42},
  number       = {2},
  year         = {2018},
  url          = {http://www.informatica.si/index.php/informatica/article/view/1442},
  bibsource    = {dblp computer science bibliography, https://dblp.org}
}

@inproceedings{Witness-Generation-Existential-CTL,
  author       = {Chuan Jiang and
                  Gianfranco Ciardo},
  editor       = {Dirk Beyer and
                  Marieke Huisman},
  title        = {Generation of Minimum Tree-Like Witnesses for Existential {CTL}},
  booktitle    = {Tools and Algorithms for the Construction and Analysis of Systems (TACAS), Part {I}},
  series       = {Lecture Notes in Computer Science},
  volume       = {10805},
  pages        = {328--343},
  publisher    = {Springer},
  year         = {2018},
  doi          = {10.1007/978-3-319-89960-2\_18},
  bibsource    = {dblp computer science bibliography, https://dblp.org}
}

@inproceedings{Tree-Like-Counterexamples,
  author       = {Edmund M. Clarke and
                  Somesh Jha and
                  Yuan Lu and
                  Helmut Veith},
  title        = {Tree-Like Counterexamples in Model Checking},
  booktitle    = {17th {IEEE} Symposium on Logic in Computer Science (LICS)},
  pages        = {19--29},
  publisher    = {{IEEE} Computer Society},
  year         = {2002},
  doi          = {10.1109/LICS.2002.1029814},
  bibsource    = {dblp computer science bibliography, https://dblp.org}
}

@inproceedings{Counterexamples-for-Classification,
  author       = {Fabio Martinelli and
                  Francesco Mercaldo and
                  Antonella Santone},
  title        = {On the Adoption of Counterexample for Classification Task Explainability},
  booktitle    = {32nd International Conference on Enabling Technologies: Infrastructure
                  for Collaborative Enterprises (WETICE)},
  pages        = {140--145},
  publisher    = {{IEEE}},
  year         = {2024},
  doi          = {10.1109/WETICE64632.2024.00033},
  bibsource    = {dblp computer science bibliography, https://dblp.org}
}

@misc{Counterexample-driven-MC,
  title={Counterexample-driven model checking},
  author={Shankar, Natarajan and Sorea, Maria},
  howpublished={SRI International},
  url={http://www.csl.sri.com/people/shankar/wmc.pdf},
  year={2003}
}

@inproceedings{Counterexamples-and-Witnesses,
  author       = {Edmund M. Clarke and
                  Orna Grumberg and
                  Kenneth L. McMillan and
                  Xudong Zhao},
  editor       = {Bryan Preas},
  title        = {Efficient Generation of Counterexamples and Witnesses in Symbolic
                  Model Checking},
  booktitle    = {Proceedings of the 32st Conference on Design Automation},
  pages        = {427--432},
  publisher    = {{ACM} Press},
  year         = {1995},
  doi          = {10.1145/217474.217565},
  bibsource    = {dblp computer science bibliography, https://dblp.org}
}

@inproceedings{CTL1,
  author       = {E. Allen Emerson and
                  Edmund M. Clarke},
  editor       = {J. W. de Bakker and
                  Jan van Leeuwen},
  title        = {Characterizing Correctness Properties of Parallel Programs Using Fixpoints},
  booktitle    = {Automata, Languages and Programming (ICALP)},
  series       = {Lecture Notes in Computer Science},
  volume       = {85},
  pages        = {169--181},
  publisher    = {Springer},
  year         = {1980},
  doi          = {10.1007/3-540-10003-2\_69},
  bibsource    = {dblp computer science bibliography, https://dblp.org}
}

@inproceedings{LTL,
  author       = {Amir Pnueli},
  title        = {The Temporal Logic of Programs},
  booktitle    = {18th Annual Symposium on Foundations of Computer Science (FoCS)},
  pages        = {46--57},
  publisher    = {{IEEE} Computer Society},
  year         = {1977},
  doi          = {10.1109/SFCS.1977.32},
  bibsource    = {dblp computer science bibliography, https://dblp.org}
}

@inproceedings{CTL2,
  author       = {Edmund M. Clarke and
                  E. Allen Emerson},
  editor       = {Dexter Kozen},
  title        = {Design and Synthesis of Synchronization Skeletons Using Branching-Time
                  Temporal Logic},
  booktitle    = {Logics of Programs},
  series       = {Lecture Notes in Computer Science},
  volume       = {131},
  pages        = {52--71},
  publisher    = {Springer},
  year         = {1981},
  doi          = {10.1007/BFB0025774},
  bibsource    = {dblp computer science bibliography, https://dblp.org}
}

@inproceedings{CTL-vs-LTL,
  author       = {Moshe Y. Vardi},
  editor       = {Tiziana Margaria and
                  Wang Yi},
  title        = {Branching vs. Linear Time: Final Showdown},
  booktitle    = {Tools and Algorithms for the Construction and Analysis of Systems (TACAS)},
  series       = {Lecture Notes in Computer Science},
  volume       = {2031},
  pages        = {1--22},
  publisher    = {Springer},
  year         = {2001},
  doi          = {10.1007/3-540-45319-9\_1},
  bibsource    = {dblp computer science bibliography, https://dblp.org}
}

@article{CTL*,
  author       = {E. Allen Emerson and
                  Joseph Y. Halpern},
  title        = {{``Sometimes'' and ``Not Never'' revisited: on branching versus linear
                  time temporal logic}},
  journal      = {J. {ACM}},
  volume       = {33},
  number       = {1},
  pages        = {151--178},
  year         = {1986},
  doi          = {10.1145/4904.4999},
  bibsource    = {dblp computer science bibliography, https://dblp.org}
}

@inproceedings{mu1,
  author       = {Dexter Kozen},
  editor       = {Mogens Nielsen and
                  Erik Meineche Schmidt},
  title        = {Results on the Propositional {\(\mathrm{\mu}\)}-Calculus},
  booktitle    = {Automata, Languages and Programming (ICALP)},
  series       = {Lecture Notes in Computer Science},
  volume       = {140},
  pages        = {348--359},
  publisher    = {Springer},
  year         = {1982},
  doi          = {10.1007/BFB0012782},
  bibsource    = {dblp computer science bibliography, https://dblp.org}
}

@article{mu2,
  author       = {Dexter Kozen},
  title        = {Results on the Propositional mu-Calculus},
  journal      = {Theor. Comput. Sci.},
  volume       = {27},
  pages        = {333--354},
  year         = {1983},
  doi          = {10.1016/0304-3975(82)90125-6},
  bibsource    = {dblp computer science bibliography, https://dblp.org}
}

@inproceedings{ALCCTL,
  author       = {Franz Weitl and
                  Shin Nakajima and
                  Burkhard Freitag},
  editor       = {Jos{\'{e}} Luiz Fiadeiro and
                  Stefania Gnesi and
                  Andrea Maggiolo{-}Schettini},
  title        = {Structured Counterexamples for the Temporal Description Logic {ALCCTL}},
  booktitle    = {8th {IEEE} International Conference on Software Engineering and Formal
                  Methods (SEFM)},
  pages        = {232--243},
  publisher    = {{IEEE} Computer Society},
  year         = {2010},
  doi          = {10.1109/SEFM.2010.36},
  bibsource    = {dblp computer science bibliography, https://dblp.org}
}

@article{ACTL-Linear-Counterexamples,
  author       = {Francesco Buccafurri and
                  Thomas Eiter and
                  Georg Gottlob and
                  Nicola Leone},
  title        = {On {ACTL} Formulas Having Linear Counterexamples},
  journal      = {J. Comput. Syst. Sci.},
  volume       = {62},
  number       = {3},
  pages        = {463--515},
  year         = {2001},
  doi          = {10.1006/JCSS.2000.1734},
  bibsource    = {dblp computer science bibliography, https://dblp.org}
}

@article{Framework-for-Counterexample-Generation,
  author       = {Marsha Chechik and
                  Arie Gurfinkel},
  title        = {A framework for counterexample generation and exploration},
  journal      = {Int. J. Softw. Tools Technol. Transf.},
  volume       = {9},
  number       = {5-6},
  pages        = {429--445},
  year         = {2007},
  doi          = {10.1007/S10009-007-0047-9},
  bibsource    = {dblp computer science bibliography, https://dblp.org}
}

@article{Review-on-Counterexample-Explanation,
  author       = {Arut Prakash Kaleeswaran and
                  Arne Nordmann and
                  Thomas Vogel and
                  Lars Grunske},
  title        = {A systematic literature review on counterexample explanation},
  journal      = {Inf. Softw. Technol.},
  volume       = {145},
  year         = {2022},
  doi          = {10.1016/J.INFSOF.2021.106800},
  bibsource    = {dblp computer science bibliography, https://dblp.org}
}

@inproceedings{Efficient-Evidence-Generation,
  author       = {Anna Stramaglia and
                  Jeroen J. A. Keiren and
                  Maurice Laveaux and
                  Tim A. C. Willemse},
  editor       = {Arie Gurfinkel and
                  Marijn Heule},
  title        = {Efficient Evidence Generation for Modal {\(\mu\)}-Calculus Model Checking},
  booktitle    = {Tools and Algorithms for the Construction and Analysis of Systems
                  (TACAS), Part {I}},
  series       = {Lecture Notes in Computer Science},
  volume       = {15696},
  pages        = {191--210},
  publisher    = {Springer},
  year         = {2025},
  doi          = {10.1007/978-3-031-90643-5\_10},
  bibsource    = {dblp computer science bibliography, https://dblp.org}
}

@phdthesis{Rich-Explanations,
  title={Symbolic Model Checking of Multi-Modal logics: Uniform Strategies and Rich Explanations},
  author={Busard, Simon},
  year={2017},
  school={Catholic University of Louvain},
  url          = {https://hdl.handle.net/2078.1/186372},
  bibsource    = {dblp computer science bibliography, https://dblp.org}
}

@inproceedings{Evidence-for-Fixpoint-Logic,
  author       = {Sjoerd Cranen and
                  Bas Luttik and
                  Tim A. C. Willemse},
  editor       = {Stephan Kreutzer},
  title        = {Evidence for Fixpoint Logic},
  booktitle    = {24th {EACSL} Annual Conference on Computer Science Logic (CSL)},
  series       = {LIPIcs},
  volume       = {41},
  pages        = {78--93},
  publisher    = {Schloss Dagstuhl - Leibniz-Zentrum f{\"{u}}r Informatik},
  year         = {2015},
  doi          = {10.4230/LIPICS.CSL.2015.78},
  bibsource    = {dblp computer science bibliography, https://dblp.org}
}

@inproceedings{Evidence-Extraction-from-PBES,
  author       = {Wieger Wesselink and
                  Tim A. C. Willemse},
  editor       = {Christoph Benzm{\"{u}}ller and
                  Jens Otten},
  title        = {Evidence Extraction from Parameterised Boolean Equation Systems},
  booktitle    = {Proceedings of the 3rd International Workshop on Automated Reasoning
                  in Quantified Non-Classical Logics (ARQNL)},
  series       = {{CEUR} Workshop Proceedings},
  volume       = {2095},
  pages        = {86--100},
  publisher    = {CEUR-WS.org},
  year         = {2018},
  url          = {https://ceur-ws.org/Vol-2095/paper6.pdf},
  bibsource    = {dblp computer science bibliography, https://dblp.org}
}

@Misc{CTLviz-artefact,
  author="Arend Rensink",
  year=2026,
  title="{CTLViz}: Visualising {CTL} Witnesses and Counterexamples",
  howpublished="Software artefact; published at Zenodo",
  publisher    = {Zenodo},
  version      = {1.2.0},
  doi          = {10.5281/zenodo.19169335},
}

@Misc{CTLviz-extended,
  author="Arend Rensink",
  year=2026,
  title="Visualising {CTL} Witnesses and Counterexamples --- Extended Version",
  howpublished="Published at arXiv",
  doi          = { },
}

@InProceedings{CTLviz-spin,
  author="Arend Rensink",
  year=2026,
  title="Visualising {CTL} Witnesses and Counterexamples",
  publisher    = {Springer},
  Series="Lecture Notes in Computer Science",
  booktitle="32nd International Symposium on Model Checking Software (SPIN)",
  editor="Vincenzo Ciancia and Arnd Hartmanns",
}

@article{Kripke,
	author = {Saul Kripke},
	journal = {Acta Philosophica Fennica},
	pages = {83--94},
	title = {Semantical Considerations on Modal Logic},
	volume = {16},
	year = {1963}
}

\begin{small}
\noindent
\textbf{Disclosure of Interests.} The
author has no competing interests to declare that are relevant to the content of this article.
\end{small}

\ifextended
\clearpage
\appendix
\crefalias{section}{appendix}
\section{Proofs of the Main Results}
\label{sec:appendix}

This appendix contains proofs of the main results in the paper. Some of the proofs use induction over the \emph{depth} of a formula, given in the form of a function $\depth\colon \Form\to \Nat$ defined by
\[ \depth:\phi \mapsto
\begin{cases}
1 & \text{if } \child(\phi)=\emptyset \\
1 + \max_{\psi\in \child(\phi)} \depth(\psi) & \text{otherwise.}
\end{cases}
\]
The depth of a set $F\subseteq \Form$ is defined by $\depth(F)=\max_{\phi\in F} \depth(\phi)$.

\propevidence*
\begin{proof}
Let $M_a=\tupof{S,C,R,L\proj{S\times \child(\phi_a)}}$. Clearly $M_a$ is consistent ($M$ is a sound supermodel), and its domain is correct by construction. Let $N\supmodel M_a$ be sound. Since $s_a\models^N \phi_a$ is completely determined by the combined satisfaction of the children of $\phi_a$ in $N$, which is faithfully recorded by both $L_M$ and $L_N$ since $M$ and $N$ are both sound, it follows that $L_N(s_a,\phi_a)=L_M(s_a,\phi_a)=b_a$, hence $a\in \hat L_N$.\qed
\end{proof}
The proof of \Cref{thm:evidence-if} (which formally states the properties of $\witness$ and $\counter$) depends on some intermediate results. First, we show that supermodels reflect and preserve certain successors and maximal paths.
\begin{lemma}\label{prop:supermodel preservation}
Let $M_1\submodel M_2$ and let $s\in S_1$.
\begin{enumerate}[smallsep]
\item $\succ_1(s)\subseteq \succ_2(s)$, and $\succ_1(s)=\succ_2(s)$ if $s\in C_1$.
\item For all $\rho\in\maxpath_1(s)$, there is a $\rho'\in \maxpath_2(s)$ such that $\rho\prfeq \rho'$, and $\rho=\rho'$ if $\rho$ ends in $C_1$.
\item For all $\rho\in\maxpath_2(s)$, there is a $\rho'\in \maxpath_1(s)$ such that $\rho'\prfeq \rho$, and $\rho'=\rho$ if $\rho'$ ends in $C_1$.
\end{enumerate}
\end{lemma}
The following states that, for full models, precisely one of $\witness$ and $\counter$ always holds (for any state and formula).
\begin{lemma}[$\witness$ and $\counter$ are complementary in full models.]\label{lem:witness-counter-complementary}
Let $F\in\CForm$ and let $M\in\Model(F)$ be full. For all $s\in S$ and $\phi\in F_\Oper$, exactly one of $\witness(s,\phi)$ and $\counter(s,\phi)$ holds.
\end{lemma}

\begin{proof}
Note that, because $M$ is full, $C=S$ and for all $\xi\in F$, $S$ is partitioned into $S\proj{\xi,\ttt}$ and $S\proj{\xi,\fff}$.

The cases of $\phi=\lneg\psi$ and $\phi= \psi_1\lor \psi_2$ are immediate. For each of the temporal operators, the complementarity of $\counter(s,\phi)$ and $\witness(s,\phi)$ comes down to two set equations, one implying that either $\witness$ or $\counter$ holds, and the other implying that they cannot both hold. Below, $\bar X$ (for $X\subseteq S$) stands for $S\setminus X$.
\begin{itemize}
\item For $\EX\psi$, setting $X=S\proj{\psi,\ttt}$ and hence $\bar X=S\proj{\psi,\fff}$:
\begin{align*}
X\cup \bar X
 & = S \\
X\cap \bar X
  & = \emptyset \enspace.
\end{align*}
\item For $\EUof{\psi_1}{\psi_2}$, setting $X=S\proj{\psi_1,\ttt}$ and $Y=S\proj{\psi_2,\ttt}$, and hence $\bar X=S\proj{\psi_1,\fff}$ and $\bar Y=S\proj{\psi_2,\fff}$:
\begin{align*}
X\finseq\cc Y\cc S\finfseq
  \cup
\left(\bar Y\finfseq\cup \bigl(\bar Y\finseq\cc (\bar X\cap \bar Y)\cc S\finfseq\bigr)\right)
 & = S\finfseq \\
X\finseq\cc Y\cc S\finfseq
  \cap 
\left(\bar Y\finfseq\cup \bigl(\bar X\finseq\cc (\bar X\cap \bar Y)\cc S\finfseq\bigr)\right)
 & = \emptyset \enspace.
\end{align*}

\item For $\EF\psi$, setting $X=S\proj{\psi,\ttt}$ and hence $\bar X=S\proj{\psi,\fff}$:
\begin{align*}
\left((X\finseq\cc X) \cup X\infseq\right)
  \cup
S\finseq \cc \bar X\cc S\finfseq
 & = S\finfseq \\
\left((X\finseq\cc X) \cup X\infseq\right)
  \cap 
S\finseq \cc \bar X\cc S\finfseq
 & = \emptyset \enspace.
\end{align*}
\end{itemize}
All of these equations are straightforward to show.\qed
\end{proof}
The following lemma states that $\witness$ and $\counter$ are preserved in supermodels.

\begin{lemma}\label{lem:witness-counter-upward-closed}
If $M\submodel N$, then $\witness_M(s,\phi)$ implies $\witness_N(s,\phi)$ and $\counter_M(s,\phi)$ implies $\counter_N(s,\phi)$.
\end{lemma}

\begin{proof}
Note that we have $S_M\proj{\xi,b}\subseteq S_N\proj{\xi,b}$ and $C_M\proj{\xi,b}\subseteq C_N\proj{\xi,b}$ for all $\xi\in F$ and $b\in\Bool$.
\begin{itemize}
\item For $\phi=\True$, $\phi=\lneg\psi$ and $\phi=\psi_1\lor \psi_2$, this follows immediately from the subset relations recalled above.

\item Assume $\phi=\EX\psi$ and $\witness_M(s,\phi)$ holds, hence $\succ_M(s)\cap S_M\proj{\psi,\ttt}\neq\emptyset$. Since $\succ_M(s)\subseteq \succ_N(s)$ (\Cref{prop:supermodel preservation}) and $S_M\proj{\psi,\ttt}\subseteq S_N\proj{\psi,\ttt}$, $\witness_N(s,\phi)$ then certainly also holds.

\item Assume $\phi=\EUof{\psi_1}{\psi_2}$ and $\witness_M(s,\phi)$ holds, hence $\maxpath_M(s)\cap S\proj{\psi_1,\ttt}\finseq\cc S\proj{\psi_2,\ttt}\cc S\finfseq \neq\emptyset$. Since $\maxpath_M(s)\subseteq \maxpath_N(s)$ (\Cref{prop:supermodel preservation}), $S_M\proj{\psi_1,\ttt} \subseteq S_N\proj{\psi_1,\ttt}$ and $S_M\proj{\psi_2,\ttt} \subseteq S_N\proj{\psi_2,\ttt}$, $\witness_N(s,\phi)$ then certainly also holds.

\item Assume $\phi=\EG\psi$ and $\witness_M(s,\phi)$ holds, hence $\maxpath_M(s)\cap \left((S\proj{\psi,\ttt}\finseq\cc C\proj{\psi,\ttt}) \cup S\proj{\psi,\ttt}\infseq\right) \neq\emptyset$. Since $\maxpath_M(s)\subseteq \maxpath_N(s)$ (\Cref{prop:supermodel preservation}), $S_M\proj{\psi,\fff} \subseteq S_N\proj{\psi,\fff}$ and $C_M\proj{\psi,\fff} \subseteq C_N\proj{\psi,\fff}$, $\witness_N(s,\phi)$ then certainly also holds.

\item Assume $\phi=\EX\psi$ and $\counter_M(s,\phi)$ holds, hence $s\in C_M$ and $\succ_M(s)\subseteq S_M\proj{\psi,\fff}$. Since $\succ_N(s)=\succ_M(s)$ (\Cref{prop:supermodel preservation}) and $S_N\proj{\psi,\fff}\supseteq S_M\proj{\psi,\fff}$, $\counter_N(s,\phi)$ then certainly also holds.

\item Assume $\phi=\EUof{\psi_1}{\psi_2}$ and $\counter_M(s,\phi)$ holds, hence $\maxpath(s)\subseteq C_M\proj{\psi_2,\fff}\finfseq\cup \bigl(C_M\proj{\psi_2,\fff}\finseq\cc (S_M\proj{\psi_1,\fff}\cap S_M\proj{\psi_2,\fff})\cc S_M\finfseq\bigr)$. We also know $C_M\proj{\psi_2,\fff}\subseteq C_N\proj{\psi_2,\fff}$, $S_M\proj{\psi_i,\fff}\subseteq S_N\proj{\psi_i,\fff}$ for $i=1,2$ and $S_M\subseteq S_N$.

Let $\rho_N\in\maxpath_N(s)$; then there is a $\rho_M\in \maxpath_M(s)$ such that $\rho_M\prfeq \rho_N$ and if $\rho_M$ ends in $C_M$, $\rho_M=\rho_N$ (\Cref{prop:supermodel preservation}). If $\rho_M\in C_M\proj{\psi_2,\fff}\finfseq$ then $\rho_M$ ends in $C_M$, hence $\rho_N=\rho_N\in C_N\proj{\psi_2,\fff}\finfseq$. Otherwise, $\rho_M\in C_M\proj{\psi_2,\fff}\finseq\cc (S_M\proj{\psi_M,\fff}\cap S_M\proj{\psi_2,\fff})\cc S_M\finfseq$ and hence $\rho_M\in C_N\proj{\psi_2,\fff}\finseq\cc (S_N\proj{\psi_1,\fff}\cap S_N\proj{\psi_2,\fff})\cc S_N\finfseq$; but then also $\rho_M\in C_N\proj{\psi_2,\fff}\finseq\cc (S_N\proj{\psi_1,\fff}\cap S_N\proj{\psi_2,\fff})\cc S_N\finfseq$. It follows that $\counter_N(s,\phi)$ holds.

\item Assume $\phi=\EG\psi$ and $\counter_M(s,\phi)$ holds, hence $\maxpath_M(s)\subseteq C_M\finseq \cc S_M\proj{\psi,\fff}\cc S_M\finfseq$. We also know $C_M\subseteq C_N$, $S_M\proj{\psi,\fff}\subseteq S_N\proj{\psi,\fff}$ and $S_M\subseteq S_N$.

Let $\rho_N\in\maxpath_N(s)$; then there is a $\rho_M\in \maxpath_M(s)$ such that $\rho_M\prfeq \rho_N$ (\Cref{prop:supermodel preservation}). Since $\rho_M\in C_M\finseq \cc S_M\proj{\psi,\fff}\cc S_M\finfseq$, we also have $\rho_M\in C_N\finseq \cc S_N\proj{\psi,\fff}\cc S_N\finfseq$ and hence $\rho_N\in C_N\finseq \cc S_N\proj{\psi,\fff}\cc S_N\finfseq$. It follows that $\counter_N(s,\phi)$ holds.\qed
\end{itemize}
\end{proof}
The proof of the theorem is now straightforward:
\thmevidenceif*
\begin{proof}
First we observe (eliding proof details) that, if $M$ is sound, $\witness(s,\phi)$ is equivalent to the condition for $s\sat^M \phi$ (for $\phi\in F_\Oper$). \Cref{lem:witness-counter-complementary} then implies that $\counter(s,\phi)$ is equivalent to $s\not\sat^M \phi$.

Now let $M$ be a consistent model, and let $N\supmodel M$ be an arbitrary sound supermodel.
\begin{enumerate}
\item If $\witness_M(s,\phi)$ holds, then (due to \Cref{lem:witness-counter-upward-closed}) also $\witness_N(s,\phi)$ holds, hence (due to the observation above) $s\models^N \phi$. It follows that $M$ is a witness for $(s,\phi)$.
\item If $\counter_M(s,\phi)$ holds, then (due to \Cref{lem:witness-counter-upward-closed}) also $\counter_N(s,\phi)$ holds, hence (due to the observation above) $s\not\models^N \phi$. It follows that $M$ is a counterexample for $(s,\phi)$.\qed
\end{enumerate}
\end{proof}
\propunconstrained*
\begin{proof}
The proof proceeds by general induction on a well-founded ordering ${\leq}\subseteq \pow\Form\times \pow\Form$ that we first have to define. For arbitrary non-empty $F_1,F_2\subseteq \Form$ let
\[ F_1\leq F_2 \;\iffdef\; \depth(F_1)<\depth(F_2) \vee (\depth(F_1)=\depth(F_2) \wedge |F_1|\leq |F_2|) \enspace.
\]
\begin{description}
\item[Base case.] Assume $G\subseteq F$ has distinct propositions and $\depth(G)=1$. Given that $\True$ is not among the allowed operators, we have $\phi=p\in \Prop$ for all $\phi\in G$, and all those $p$ are distinct. Any $M$ with $\dom L\subseteq S\times G$ is consistent: a sound supermodel is obtained by taking any completion of $L$ to a full function $L'\of S\times F_\Prop\to \Bool$ and defining $M'=\tupof{S,S,R,L'}$, and then defining $L''\of S\times F\to \Bool$ by
\[ L'':(s,\phi)\mapsto \begin{cases} \ttt & \text{if } s\sat^{M'} \phi \\ \fff & \text{otherwise.}
\end{cases} \]
Clearly, $M''=M'\sqcup\tupof{L''}$ is sound.

\item[Induction step.] Assume that $H\subseteq F$ has distinct propositions, $\depth(H)>1$ and all $G<H$ with distinct propositions are unconstrained (induction hypothesis). Let $M$ be arbitrary with $\dom L\subseteq S\times H$. We show that $M$ is consistent.

Let $\phi\in H$ be such that $\depth(\phi)=\depth(H)$ and let $G= (H\setminus \setof{\phi}) \cup \child(\phi)$; then $G<H$ and $G$ has distinct propositions, hence (by the induction hypothesis) $G$ is unconstrained. Let $L'=L\proj{S\times G}$ and $M'=\tupof{S,C,R,L'}$. It follows that $M=M'\sqcup\tupof{L}$.

We proceed by a case distinction on $\phi$. Due to $\depth(\phi)>1$ we know $\phi\in \Form_\Oper$.
\begin{itemize}
\item $\phi=\psi_1\lor \psi_2$. Let
\[ L''=L'\cup\gensetof{(s,\psi_i)\mapsto b}{L(s,\phi)=b, i=1,2} \]
For any sound supermodel $N\supmodel\tupof{L''}$, it follows that $L_N\supseteq L$; hence $N\supmodel\tupof{L}$.

\item $\phi=\lneg\psi$. Let
\[ L''=L'\cup\gensetof{(s,\psi)\mapsto\ttt}{L(s,\phi)=\fff} \cup \gensetof{(s,\psi)\mapsto \fff}{L(s,\phi)=\ttt} \]
For any sound supermodel $N\supmodel\tupof{L''}$, it follows that $L_N\supseteq L$; hence $N\supmodel\tupof{L}$.

\item $\phi=\EUof{\psi_1}{\psi_2}$. Let
\[ L''=L'\cup\gensetof{(s,\psi_2)\mapsto\ttt}{L(s,\phi)=\ttt} \cup \gensetof{(s,\psi_i)\mapsto \fff}{L(s,\phi)=\fff, i=1,2} \]
For any sound supermodel $N\supmodel\tupof{L''}$, it follows that $L_N\supseteq L$; hence $N\supmodel\tupof{L}$.
\end{itemize}
In each of these cases, since $G$ is unconstrained, $M''=M'\sqcup\tupof{L''}$ is consistent, meaning that there is a sound supermodel $N\supmodel M''$. Then $N\supmodel \tupof{L''}$ and hence (for all cases of $\phi$) $N\supmodel \tupof L$, which in turn implies $N\supmodel M'\sqcup \tupof L=M$. It follows that $N$ is a sound supermodel of $M$.\qed
\end{description}
\end{proof}
\thmnecessary*
\begin{proof}
Let $b\in\Bool$ such that $b=\ttt$ when we are proving Clause~1 and $\phi$ does not equal $\lneg\psi$ or we are proving Clause~2 and $\phi$ equals $\lneg\psi$, and $b=\fff$ otherwise, and let $\bar b$ be the complement of $b$. For all $s\in S\setminus C$, let $\hat s\notin S$ be a fresh state. We define
\begin{align*}
\hat S & = S\cup\gensetof{\hat s}{s\in S\setminus C} \\
\hat R & = R\cup\gensetof{(s,\hat s)}{s\in S\setminus C} \\
\hat L & = L\cup\gensetof{((s,\psi)\mapsto \bar b)}{s\in \hat S, \psi\in\child(\phi), (s,\psi)\notin\dom L} \\
\hat M & = \tupof{\hat S,\hat S,\hat R,\hat L} \enspace.
\end{align*}
Clearly, $\hat M\supmodel M$. Because $\child(\phi)$ is unconstrained, $\hat M$ is consistent; let $M'\supmodel \hat M$ be a sound supermodel. Since we are only interested in properties of states in $S$ and none of the states in $S'\setminus \hat S$ can be reachable from $S$, w.l.o.g.\ assume $S'=\hat S$ (implying $C'=\hat C$ and $R'=\hat R$, hence the only component that differs between $\hat M$ and $M'$ is the labelling, $\hat L\subseteq L'$). By construction, any path $\rho$ in $M'$ starting from $s$ such that $[\rho']\subseteq S$ must consist only of transitions in $R$, hence be a path in $M$. Furthermore, $S'\proj{\psi,b}=S\proj{\psi,b}$ for all $\psi\in\child(\phi)$.

Consider an arbitrary $s\in S$.
\begin{enumerate}
\item Since $M$ is a witness for $(s,\phi)$, we have $s\sat^{M'} \phi$, hence (due to \Cref{thm:evidence-if}) $\counter'(s,\phi)$ cannot hold. \Cref{lem:witness-counter-complementary} then implies that $\witness'(s,\phi)$ holds. We prove by case distinction on $\phi$ that $\witness(s,\phi)$ then also holds.
\begin{itemize}
\item $\phi=\ttt$. $\witness(s,\phi)$ holds by definition.
\item $\phi=\lneg\psi$. $\witness(s,\phi)$ holds due to $s\in S'\proj{\psi,\fff}= S\proj{\psi,\fff}$.
\item $\phi=\psi_1\lor\psi_2$. $\witness(s,\phi)$ holds due to $s\in S'\proj{\psi_1,\ttt}\cup S'\proj{\psi_2,\ttt}= S\proj{\psi_1,\ttt}\cup S\proj{\psi_2,\ttt}$.
\item $\phi=\EX\psi$. Let $s'\in \succ'(s)\cap S'\proj{\psi,\ttt}$; then $s'\in S\proj{\psi,\ttt}$, hence also $(s,s')\in R$. It follows that $\witness(s,\phi)$ holds.

\item $\phi=\EUof{\psi_1}{\psi_2}$. Let $\rho'\in \maxpath'(s)\cap (S'\proj{\psi_1,\ttt}\finseq \cc S'\proj{\psi_2,\ttt}\cc {S'}\finfseq)$. Due to \Cref{prop:supermodel preservation} there is some $\rho\prfeq \rho'$ such that $\rho\in\maxpath(s)$. Let $\rho''\prfeq \rho'$ be the prefix in $S'\proj{\psi_1,\ttt}\finseq \cc S'\proj{\psi_2,\ttt} = S\proj{\psi_1,\ttt}\finseq \cc S\proj{\psi_2,\ttt}$ Since $[\rho'']\subseteq S$, all transitions are in $R$, implying $\rho''\prfeq \rho$. It follows that $\rho\in S\proj{\psi_1,\ttt}\finseq \cc S\proj{\psi_2,\ttt}\cc S\finfseq$. We may conclude that $\witness(s,\phi)$ holds.

\item $\phi=\EG\psi$. Let $\rho'\in \maxpath'(s)\cap ((S'\proj{\psi,\ttt}\finseq \cc C'\proj{\psi,\ttt})\cup S'\proj{\psi,\ttt}\infseq)$. It follows that $\rho'\in S\proj{\psi,\ttt}\finfseq$, hence $[\rho']\subseteq S$, implying $\rho'\in\maxpath(s)$. There are two cases.
\begin{itemize}
\item $\rho'\in S'\proj{\psi,\ttt}\finseq \cc C'\proj{\psi,\ttt}$; then $\rho'=s_1\cdots s_n\in S\proj{\psi,\ttt}\finseq$. Since $\rho'$ is maximal, it follows that $\succ'(s_n)=\emptyset$, hence $s_n\in C\proj{\psi,\ttt}$.
\item $\rho'\in S'\proj{\psi,\ttt}\infseq$; then also $\rho'\in S\proj{\psi,\ttt}\infseq$.
\end{itemize}
In either case, we may conclude that $\witness(s,\phi)$ holds.
\end{itemize}

\item Since $M$ is a counterexample for $(s,\phi)$, we have $s\not\sat^{M'} \phi$, hence (due to \Cref{thm:evidence-if}) $\witness'(s,\phi)$ cannot hold. \Cref{lem:witness-counter-complementary} then implies that $\counter'(s,\phi)$ holds. We prove by case distinction on $\phi$ that $\counter(s,\phi)$ then also holds.

\begin{itemize}
\item $\phi=\ttt$. Contradiction ($\counter'(s,\phi)$ cannot hold).
\item $\phi=\lneg\psi$. $\counter(s,\phi)$ holds due to $s\in S'\proj{\psi,\ttt}= S\proj{\psi,\ttt}$.
\item $\phi=\psi_1\lor\psi_2$. $\counter(s,\phi)$ holds due to $s\in S'\proj{\psi_1,\fff}\cap S'\proj{\psi_2,\fff}= S\proj{\psi_1,\fff}\cap S\proj{\psi_2,\fff}$.
\item $\phi=\EX\psi$. $\counter'(s,\phi)$ implies $\succ(s')\subseteq S'\proj{\psi,\fff}=S\proj{\psi,\fff}$, which in turn implies $s\in C$ as well as $\succ(s)=\succ'(s)$; hence $\counter(s,\phi)$ holds.

\item $\phi=\EUof{\psi_1}{\psi_2}$. Let $\rho\in\maxpath(s)$; then $\rho\prfeq \rho'$ for some $\rho'\in \maxpath'(s)$ such that $\rho'=\rho$ if $\rho$ ends in $C$ (\Cref{prop:supermodel preservation}).

\smallskip
First assume that $s_1\cdots s_n\prfeq \rho'$ such that $s_i\in C$ for all $1\leq i<n$ and $s_n\notin C$ (which ensures $s_1\cdots s_n\prfeq \rho$), and let $\rho''=s_1\cdots s_n\cc \hat s_n$; then $\rho''\in\maxpath'(s)$. Since $\counter'(s,\phi)$ holds and $\rho''\notin C'\proj{\psi_2,\fff}\finfseq$ (since that would imply $[\rho'']\subseteq S$, contradicting $\hat s_n\notin S$), it must be the case that $\rho'' \in C'\proj{\psi_2,\fff}\infseq \cc (S'\proj{\psi_1,\fff}\cap S'\proj{\psi_2,\fff})\cc {S'}\finfseq$. Let $s_1\cdots s_i$ (with $i\leq n$) be the prefix of $\rho''$ (and hence also of $\rho$) that is in $C'\proj{\psi_2,\fff}\infseq \cc (S'\proj{\psi_1,\fff}\cap S'\proj{\psi_2,\fff})$; then clearly also $s_1\cdots s_i\in C\proj{\psi_2,\fff}\infseq \cc (S\proj{\psi_1,\fff}\cap S\proj{\psi_2,\fff})$. It follows that $\rho\in C\proj{\psi_2,\fff}\infseq \cc (S\proj{\psi_1,\fff}\cap S\proj{\psi_2,\fff})\cc S\finfseq$.

\smallskip
Otherwise, it follows that $[\rho']\subseteq C$. In that case, $\rho'=\rho$, and $\rho'\in C'\proj{\psi_2,\fff}\finfseq \cup C'\proj{\psi_2,\fff}\finseq \cc (S'\proj{\psi_1,\fff}\cap S'\proj{\psi_2,\fff}) \cc {S'}\finfseq$ immediately guarantees $\rho\in C\proj{\psi_2,\fff}\finfseq \cup C\proj{\psi_2,\fff}\finseq \cc (S\proj{\psi_1,\fff}\cap S\proj{\psi_2,\fff}) \cc S\finfseq$

\smallskip
We may conclude that $\counter(s,\phi)$ holds.

\item $\phi=\EG\psi$. Let $\rho\in\maxpath(s)$; then $\rho\prfeq \rho'$ for some $\rho'\in \maxpath'(s)$ (\Cref{prop:supermodel preservation}).

\smallskip
First assume that $s_1\cdots s_n\prfeq \rho'$ such that $s_i\in C$ for all $1\leq i<n$ and $s_n\notin C$ (which ensures $s_1\cdots s_n\prfeq \rho$), and let $\rho''=s_1\cdots s_n\cc \hat s_n$; then $\rho''\in\maxpath'(s)$. Since $\counter'(s,\phi)$ holds, we have $\rho''\in {C'}\finseq\cc S'\proj{\psi,\fff}\cc {S'}\finfseq$. Let $s_1\cdots s_i$ (with $i\leq n$) be the prefix of $\rho''$ (and hence also of $\rho$) that is in ${C'}\finseq\cc S'\proj{\psi,\fff}$; then clearly also $s_1\cdots s_i\in C\finseq\cc S\proj{\psi,\fff}$. It follows that $\rho \in C\finseq\cc S\proj{\psi,\fff}\cc S\finfseq$.

\smallskip
Otherwise, it follows that $[\rho']\subseteq C$. In that case, $\rho'=\rho$, and $\rho'\in {C'}\finseq\cc S'\proj{\psi,\fff}\cc {S'}\finfseq$ immediately guarantees $\rho\in C\finseq\cc S\proj{\psi,\fff}\cc S\finfseq$.

\smallskip
We may conclude that $\counter(s,\phi)$ holds.\qed
\end{itemize}
\end{enumerate}
\end{proof}
\thmminevidence*
\begin{proof}
The proof consists of two parts: showing that $\minwitness_M(s,\phi)$ implies $\witness_M(s,\phi)$, and showing that if $\minwitness_M(s,\phi)$ holds and $N\propersubmodel M$, $\witness_N(s,\phi)$ does \emph{not} hold --- and analogously for $\mincounter$ and $\counter$. We will prove this for all $N$ that are \emph{direct} predecessors of $M$ in the $\propersubmodel$-relation, i.e., for which there is no intermediate $N'$ with $N\propersubmodel N'\propersubmodel M$.  (Due to \Cref{lem:witness-counter-upward-closed}, this suffices.) We denote $N\directsubmodel M$ in this case. This implies that one of the components $S_N$, $C_N$, $R_N$, or $L_N$ is just one element smaller than its counterpart in $M$, and the others are equal.
\begin{enumerate}[topsep=\smallskipamount]
\item By a case distinction on $\phi$.
\begin{itemize}
\item $\phi=\ttt$. $\witness_M(s,\phi)$ holds always, and $M$ is the smallest model that includes $s$ so there does not exist any $N\submodel M$.

\item $\phi=\lneg\psi$. $\witness_M(s,\phi)$ clearly holds, and the only $N\propersubmodel M$ has $L_N=\emptyset$, hence $\witness_N(s,\phi)$ does not hold.

\item $\phi=\psi_1\lor\psi_2$. $\witness_M(s,\phi)$ clearly holds, and the only $N\propersubmodel M$ has $L_N=\emptyset$, hence $\witness_N(s,\phi)$ does not hold.

\item $\phi=\EX\psi$. $\succ_M(s)=\setof{s'}$ and $s'\in S_M\proj{\psi,\ttt}$ so $\witness_M(s,\phi)$ holds. If $N\directsubmodel M$ then either $R_N=\emptyset$, hence $\succ_N(s)=\emptyset$, or $L_N=\emptyset$, hence $S_N\proj{\psi,\ttt}=\emptyset$. In either case, $\witness_N(s,\phi)$ does not hold.

\item $\phi=\EUof{\psi_1}{\psi_2}$. $\maxpath_M(s)=\setof\rho$ with $\rho\in S_M\proj{\psi_1,\ttt}\finseq\cc S_M\proj{\psi_2,\ttt}$, hence $\witness_M(s,\phi)$ holds. If $N\directsubmodel M$ then either $R_N\subset R_M$, in which case $\maxpath_N(s)=\setof{\rho'}$ for some $\rho'\prf \rho$, which then does not end on a state in $S_N\proj{\psi_2,\ttt}$; or $R_N=R_M$ but $L_N\subset L_M$, in which case $\rho\notin S_N\proj{\psi_1,\ttt}\finseq\cc S_N\proj{\psi_2,\ttt}$. In either case, $\witness_N(s,\phi)$ does not hold.

\item $\phi=\EG\psi$. $\maxpath_M(s)=\setof\rho$ with $\rho\in S_M\proj{\psi,\ttt}\finseq\cc C_M\proj{\psi,\ttt}$ if $<_\rho$ has a maximum (``first kind''), or $\rho\in S_M\proj{\psi,\ttt}\infseq$ and $C=\emptyset$ otherwise (``second kind''). If $N\directsubmodel M$, then either $C_N\subset C_M$ and $R_N=R_M$, which means that $\rho$ was of the first kind but does not end in $C_N$; or $C_N=C_M$ and $R_N\subset R_M$, in which case $\maxpath_N(s)=\setof{\rho'}$ for some $\rho'\prf\rho$, which is then certainly finite and does not end in $C_N$; or $C_N=C_M$, $R_N=R_M$ and $L_N\subset L_M$, in which case $\maxpath_N(s)=\setof\rho$ but $\rho\notin S_N\proj{\psi,\ttt}\finfseq$. In all cases, $\witness_N(s,\phi)$ does not hold.
\end{itemize}
\item By a case distinction on $\phi$.
\begin{itemize}
\item $\phi=\ttt$. Contradiction.

\item $\phi=\lneg\psi$. $\counter_M(s,\phi)$ clearly holds, and the only $N\propersubmodel M$ has $L_N=\emptyset$, hence $\counter_N(s,\phi)$ does not hold.

\item $\phi=\psi_1\lor\psi_2$. $\counter_M(s,\phi)$ clearly holds, and if $N\propersubmodel M$ then $|L_N|=1$, hence $\counter_N(s,\phi)$ does not hold.

\item $\phi=\EX\psi$. $s\in C_M$ and $\succ(s)\subseteq S_M\proj{\psi,\fff}$, so $\counter_M(s,\phi)$ holds. If $N\directsubmodel M$ then $S_M=S_N$ (otherwise $R_M=R_N$ could not hold), so either $C_N=\emptyset$ (in which case $s\notin C_N$) or $C_N=C_M$ and $R_N\subset R_M$ (which contradicts the conditions for $N\propersubmodel M$) or $R_N=R_M$ and $L_N\subset L_M$ (in which case $\succ_N(s)=\succ_M(s)\not\subseteq S_N\proj{\psi,\fff}$). In all these cases, $\counter_N(s,\phi)$ does not hold.

\item $\phi=\EUof{\psi_1}{\psi_2}$. For all $\rho\in\maxpath_M(s)$, either $\rho\in C_M\proj{\psi_1,\fff}\finfseq$ or $\rho\in S_M\proj{\psi_1,\fff}\finseq$ with $\rho\last\notin C_M$; in the latter case, $\rho\last\in S_M\proj{\psi_2,\fff}$ and $\rho\proj i\in C_M$ for all $i<|\rho|$. It follows that $\counter_M(s,\phi)$ holds. Now assume $N\directsubmodel M$. Since all of $S_M$ is reachable from $s$, $S_N\subset S_M$ contradicts $R_N=R_M$. We proceed by a case distinction on the conditions for $N\directsubmodel M$.
\begin{itemize}
\item $C_N\subset C_M$, $R_N=R_M$ and $L_N=L_M$; then $\maxpath_N(s)=\maxpath_M(s)$. Let $\rho\in \maxpath_N(s)$ such that $\rho_i\in C_M\setminus C_N$ for some $1\leq i\leq |\rho|$. If $i<|\rho|$ then $\counter_N(s,\phi)$ fails at once. Otherwise, $\rho\last\in C_M$, which means $\rho\last\notin S_M\proj{\psi_2,\fff} = S_N\proj{\psi_2,\fff}$; again, $\counter_N(s,\phi)$ fails for $\rho$.

\item $C_N=C_M$, $R_N\subset R_M$ and $L_N=L_M$. Since all transitions in $R_M$ start in a state in $C_M$, this contradicts the conditions for $N\submodel M$.

\item $C_N=C_M$, $R_N=R_M$ and $L_N\subset L_M$. This means that either there is some $s'\in {S_M}\proj{\psi_1,\fff}\setminus S_N\proj{\psi_1,\fff}$ or there is some $s'\in S_M\proj{\psi_2,\fff}\setminus S_N\proj{\psi_2,\fff}$ (implying $s\notin C_M$). Since all $s'\in S_M$ are reachable from $s$, there is a $\rho\in\maxpath_N(s)$ that includes $s'$; this $\rho$ then no longer satisfies the conditions of $\counter_N(s,\phi)$.
\end{itemize}

\item $\phi=\EG\psi$. For all $\rho\in\maxpath_M(s,\phi)$, it follows that $\rho\in S_M\finseq$ (because $R_M$ is acyclic), $\rho\proj i\in C_M$ for all $i<|\rho|$ and $\rho\last\in S_M\proj{\psi,\fff}$; hence $\counter_M(s,\phi)$ holds. Now assume $N\directsubmodel M$. Since all of $S_M$ is reachable from $s$, $S_N\subset S_M$ contradicts $R_N=R_M$. We proceed by a case distinction on the conditions for $N\directsubmodel M$.
\begin{itemize}
\item $C_N\subset C_M$, $R_N=R_M$ and $L_N=L_M$; then $\maxpath_N(s)=\maxpath_M(s)$. Let $\rho\in \maxpath_N(s)$ such that $\rho_i\in C_M\setminus C_N$ for some $1\leq i\leq |\rho|$. Since $C\cap\max R=\emptyset$, it must be the case that $i<|\rho|$; but then $\rho$ fails to satisfy the conditions of $\counter_N(s,\phi)$.
\item $C_N=C_M$, $R_N\subset R_M$ and $L_N=L_M$. Since all transitions in $R_M$ start in a state in $C_M$, this contradicts the conditions for $N\submodel M$.
\item $C_N=C_M$, $R_N=R_M$ and $L_N\subset L_M$. This means there is some $s'\in \max R_N\setminus S_N\proj{\psi,fff}$. Any $\rho\in \maxpath_N(s)$ that ends on $s'$ (of which is at least one) fails to satisfy the conditions of $\counter_N(s,\phi)$.\qed
\end{itemize}
\end{itemize}
\end{enumerate}
\end{proof}
\propproof*
\begin{proof}
Let $F\in \CForm$.
\begin{description}[topsep=\smallskipamount]
\item[If.] Let $\tupof{N,\pi}$ be a proof over $F$. We show that $N$ is consistent; this then naturally also holds for any $M\submodel N$.

Let $n=\depth(F)$ and for $0\leq i\leq n$ let $M_i\submodel M$ with $S_i=S$, $C_i=C$, $R_i=R$ and $\dom L_i=\dom L\cap (S\times \gensetof{\phi \in F_\Oper}{\depth(\phi)\leq i})$. Clearly $M= M_n$. We show by induction that all $M_i$ are consistent.
\begin{description}
\item[Base case.] If $i=0$ then $\dom L_i=S\times F_\Prop$; since $F_\Prop$ is unconstrained, $M_i$ is consistent.
\item[Induction step.] Assume $M_{i-1}$ is consistent, and let $N\supmodel M_{i-1}$ be a sound supermodel. For all $a=(s,\phi,b)\in \hat L_i$ we have $\depth(\child(\phi))<i$, hence $P(a)\submodel M_{i-1}$. It follows that $N\supmodel P(a)$. Since $P(a)$ is evidence for $a$, it follows that $a\in \hat L_N$. But then also $M_i\submodel N$, hence $M_i$ is consistent.
\end{description}

\item[Only if.] Let $M\in\Model(F)$ be consistent, and let $N\supmodel M$ be sound. For all $a\in \hat {L_N}_\Oper$ let $\pi(a)$ be defined by $\tupof{S,C,R,L\cap(S\times \child(\phi_a))}$. Clearly, $\pi(a)\submodel M$ and according to \Cref{prop:evidence}, $\pi(a)$ is evidence for $a$. It follows that $P=\tupof{N,\pi}$ is the required proof.\qed
\end{description}
\end{proof}

\fi
\end{document}
% The following is recommended on https://trevorcampbell.me/html/arxiv.html
\typeout{get arXiv to do 4 passes: Label(s) may have changed. Rerun}